\documentclass[11pt,a4paper]{article}

\usepackage{styling}
\usepackage[utf8]{inputenc}

\hypersetup{
  pdftitle={Balanced Allocations in Batches: The Tower of Two Choices},
  pdfauthor={Dimitrios Los, Thomas Sauerwald}
  }

\allowdisplaybreaks

\begin{document}

\title{Balanced Allocations in Batches: The Tower of Two Choices\footnote{Full version of a paper appearing in SPAA 2023.}}
\author[]{Dimitrios Los\thanks{\texttt{dimitrios.los@cl.cam.ac.uk}} }
\author[]{Thomas Sauerwald\thanks{\texttt{thomas.sauerwald@cl.cam.ac.uk}}}
\affil[]{Department of Computer Science \& Technology, University of Cambridge, UK}

\date{} %

\maketitle

\begin{abstract}
In the balanced allocation framework, the goal is to allocate $m$ balls into $n$ bins, so as to minimize the gap (difference of maximum to average load). The \OneChoice process allocates each ball to a bin sampled independently and uniformly at random. 
The \TwoChoice process allocates balls sequentially, and each ball is placed in the least loaded of two sampled bins.
Finally, the \OnePlusBeta-process mixes these processes, meaning each ball is allocated using \TwoChoice with probability $\beta \in (0,1)$, and using \OneChoice otherwise.

Despite \TwoChoice being optimal in the sequential setting, 
it has been observed in practice that it does not perform well 
in a parallel environment, where load information may be outdated.
Following~\cite{BCEFN12}, we study such a parallel setting where balls are allocated in batches of size $b$, and balls within the same batch are allocated with the same strategy and based on the same load information.

For small batch sizes $b \in [n,n \log n]$, it was shown in \cite{LS22Noise} that \TwoChoice achieves an asymptotically optimal gap among all allocation processes with two (or any constant number of) samples.

In this work, we focus on larger batch sizes $b \in [n \log n,n^3]$. It was proved in~\cite{LS22Batched} that \TwoChoice leads to a gap of $\Theta(b/n)$. As our main result, we prove that the gap reduces to 
$\Oh(\sqrt{(b/n) \cdot \log n})$, if one runs the $(1+\beta)$-process with an appropriately chosen $\beta$ (in fact this result holds for a larger class of processes). 
This not only proves the phenomenon that \TwoChoice is not the best (leading to the formation of ``towers'' over previously light bins), but also that mixing two processes (\OneChoice and \TwoChoice) leads to a process which achieves a gap that is asymptotically smaller than both. We also derive a matching lower bound of $\Omega(\sqrt{(b/n) \cdot \log n})$ for any allocation process, which demonstrates that the above $(1+\beta)$-process is asymptotically optimal.

Our analysis also 
works in the presence of randomly weighted balls, and also implies exponential tails for the number of bins above a certain load value.

\end{abstract}

\clearpage

\clearpage
\tableofcontents
~
\clearpage

\section{Introduction}

\paragraph{Sequential balanced allocations.} In the sequential balanced allocations framework, there are $m$ tasks (balls) to be allocated into $n$ servers (bins). It is well-known that allocating the balls into bins sampled uniformly at random (a.k.a.~\OneChoice) leads \Whp\footnote{In general, with high probability refers to probability of at least $1 - n^{-c}$ for some constant $c > 0$.}~to a maximum load of $\Theta(\log n/\log \log n)$ for $m = n$ and a gap (maximum load minus average load) of $\Theta\big(\sqrt{(m/n) \cdot \log n}\big)$ for $m \geq n \log n$. 

An improvement over \OneChoice is the \DChoice process~\cite{KLM96,ABKU99,BCSV06}, where each ball is allocated to the least loaded of $d$ bins sampled uniformly at random. For any $m \geq n$, this process achieves \Whp~an $\log_d \log n + \Theta(1)$ gap, i.e., a gap that does not depend on $m$. For $d = 2$, this great improvement is known as ``power-of-two-choices'' (see also surveys~\cite{MRS01,W17} for more details). Despite the simplistic nature of the balanced allocation framework, the \TwoChoice process has had a significant impact on practical applications such as load balancing and distributed storage systems, which was also acknowledged by the ``\emph{ACM Paris Kanellakis Theory and Practice Award 2020}''~\cite{award20} (see also \textit{Applications} below).

Several variants of \TwoChoice have been studied. Of particular importance to this work is the \OnePlusBeta-process, where each ball is allocated using \TwoChoice with probability $\beta \in (0, 1]$ and \OneChoice otherwise. Mitzenmacher~\cite[Section 4.4.1]{M96} introduced this process as a model of~\TwoChoice with erroneous comparisons. Peres, Talwar and Wieder~\cite{PTW15} showed that for $\beta := \beta(n) \ll 1$, it achieves \Whp~a $\Theta((\log n)/\beta)$ gap (see also~\cite{LS22Batched}), which becomes worse for smaller $\beta$, but still remains independent of $m$. The \OnePlusBeta-process has also been applied to the analysis of \TwoChoice in the popular \emph{graphical setting}~\cite{KP06,BF21,PTW15}, where bins are organized as vertices in a graph, and each ball is allocated to the lesser loaded of two adjacent vertices of an edge sampled uniformly at random.

Another variant of \TwoChoice that has received some attention recently is the family of \TwoThinning processes~\cite{FG18,FL20}, where the ball is allocated to the second sample only if the first one does not meet a certain criterion, e.g., based on a threshold on its load or a quantile on its rank.

It should be noted that the analyses of all these processes strongly rely on the fact that the load information of each bin is updated after each allocation. In effect this means balls can only be allocated sequentially, which is a downside in distributed and parallel environments.

\paragraph{Outdated information settings.} In this work, we demonstrate that in outdated information settings by choosing an appropriately small $\beta$, \OnePlusBeta achieves the asymptotically optimal gap among a large class of processes, including not only \TwoChoice (and \OneChoice), but even adaptive processes that may allocate with a different scheme after each batch. This confirms earlier empirical observations that the performance of the \TwoChoice process deteriorates under outdated information and delays~\cite{W86,M00,D00,OWZS13,FGCBG97}. %

Berenbrink, Czumaj, Englert, Friedetzky and Nagel~\cite{BCEFN12} introduced the \Batched setting where balls are allocated in batches of size $b$. That means, in every batch the $b$ balls are allocated \textit{in parallel}, as the decision where to allocate the ball only depends on the load configuration before that batch of balls arrived. For $b = n$, they proved that \TwoChoice achieves \Whp~an $\Oh(\log n)$ gap. This bound was recently improved to $\Theta(\log n/\log \log n)$ in~\cite{LS22Noise}, and in the same work, it was shown that \TwoChoice has a gap that matches the maximum load of \OneChoice for $b$ balls, for any batch size $b \in [n \cdot e^{-\log^{\Theta(1)} n}, n \log n]$, and so it is asymptotically optimal. In contrast, for $b \geq n \log n$, \TwoChoice (and a family of other processes) have \Whp~a $\Theta(b/n)$ gap~\cite{LS22Batched}, a bound which was shown to hold even in the presence of weights and on some graphs. This analysis also demonstrates that increasing $d$ in the \DChoice process, does not always improve the gap, which is in sharp contrast to the sequential setting. In~\cite{LS22Noise}, a more powerful setting, \TauDelay was studied for the \TwoChoice process, where an adversary can choose to report for each of the bins any load from the last $\tau$ steps. For $b = \tau$, \Batched is a special instance of \TauDelay and for any $\tau \leq n \log n$, the same asymptotic bounds where shown to hold.

Outdated information settings have also been studied in the queuing setting~\cite{W86,AN92,KK95,FGCBG97,M00}. In particular, Mitzenmacher~\cite{M00} studied the corresponding version of the \Batched setting, called the \textit{bulletin board model with periodic updates}, showing that some processes requiring centralized coordination can outperform \TwoChoice, but no explicit rigorous bounds were proven. This shortcoming of \TwoChoice was characterized as \textit{herd behavior}, meaning that some of the initially lighter bins receive disproportionately many balls, turning them into heavy bins. In another empirically study, Dahlin~\cite{D00} also observed the herd behavior and suggested similar centralized strategies to improve upon \DChoice. Regarding identifying optimal processes, Whitt~\cite{W86} remarks:
\begin{quote}
\textit{We have shown that several natural selection rules are not optimal in various situations, but we have not identified any optimal rules. Identifying optimal rules in these situations would obviously be interesting, but appears to be difficult. Moreover, knowing an optimal rule might not be so useful because the optimal rule may be very complicated.}
\end{quote}

\paragraph{Applications.} Recently, several \textit{distributed low-latency schedulers}, including Sparrow~\cite{OWZS13}, Eagle~\cite{DDDZ16}, Hawk~\cite{DDKZ15}, Peacock~\cite{KG18}, Pigeon~\cite{WLLSRCJ19} and Tarcil~\cite{DSK15}, have used variants of the \TwoChoice process.  In \cite{OWZS13}, with regards to the implementation of Sparrow, the authors state: 
\begin{quote}
\textit{The power of two choices suffers from
two remaining performance problems: first, server queue
length is a poor indicator of wait time, and second, due
to messaging delays, multiple schedulers sampling in
parallel may experience race conditions.}
\end{quote}
Similar observations have been made in the context of \textit{distributed stream processing}~\cite{NMGKS15,NMKS16} and \textit{load balancers}~\cite{LXKGLG11}. These studies support that batch sizes $b = \Omega(n \log n)$ for which \TwoChoice is no longer optimal are relevant to real-world applications.

\paragraph{Weighted settings.} Several works study balanced allocation processes with \textit{weights}~\cite{TW07,BFHM08,PTW15,LS22Batched}. We will be focusing on weights sampled independently from probability distributions with bounded moment generating functions as in~\cite{LS22Batched} and~\cite{PTW15}, which includes the geometric, exponential and Poisson distributions.

\paragraph{Our results.} In this work, we prove that a family of processes satisfying a mild technical condition achieve the asymptotically optimal gap\footnote{By \textit{optimal} we mean over all processes that choose a probability allocation vector $p$, where $p_i$ gives the probability to allocate to the $i$-th heaviest bin, at the beginning of the batch and this vector remains the same throughout the entire batch.
} of $\Oh\big(\sqrt{(b/n) \cdot \log n} \big)$ in the weighted \Batched setting for $b \in [2n \log n, n^3]$, leading to roughly a quadratic improvement over the gap of the \TwoChoice process. This family of processes includes the \OnePlusBeta-process, which is a process that can be easily implemented in a \textit{decentralized manner}, and demonstrates that by setting $\beta = \sqrt{(n/b) \cdot \log n}$ we attain this asymptotically optimal gap. %

We also provide lower bounds establishing the tightness of our upper bounds. Interestingly, the lower bound of $\Omega( \sqrt{(b/n) \cdot \log n} )$ applies to a much more powerful class of allocation processes, where the allocation rule is arbitrarily tailored at the beginning of the batch.

The intuition for these optimal processes relates to the herd behavior observed in~\cite{M00} and \cite{D00}. For the \DChoice process, the maximum probability of allocating to a bin is $\max_{i \in [n]} p_i \approx d/n$. This means that, for example, in \TwoChoice in a batch of $b$ balls there are some bins that receive $\approx 2b/n$ balls and so a gap of $\approx b/n$ arises. This becomes worse as $d$ grows. To avoid this, we will investigate processes where $\max_{i \in [n]} p_i = (1 + o(1))/n$, which means that in expectation no bin receives too many balls in any particular batch. For example, the \OnePlusBeta-process has $\max_{i \in [n]} p_i \approx (1 + \beta)/n$, which means that this mixing of \OneChoice steps with \TwoChoice steps circumvents the herd behavior. See \cref{fig:two_choie_vs_batch_visual} for a visualization of how \OnePlusBeta achieves a more balanced distribution than \TwoChoice over one batch, and \cref{fig:batched_unit_weights_high_level} for how the gaps of different processes are getting worse with larger $\max_{i \in [n]} p_i$. The asymptotic gap bounds of the \OneChoice, \TwoChoice and \OnePlusBeta processes in the \Batched setting are summarized in \cref{tab:gap_summary}. Our results also imply bounds for the shape of the load vector (see 
\cref{rem:load_vector_shape}). Our analysis also 
applies in the presence of randomly weighted balls, and also implies exponential tails for the number of bins above a certain load value.

\begin{figure}
    \centering
\begin{minipage}[t]{0.3\textwidth}
\begin{center}
\includegraphics[scale=0.15]{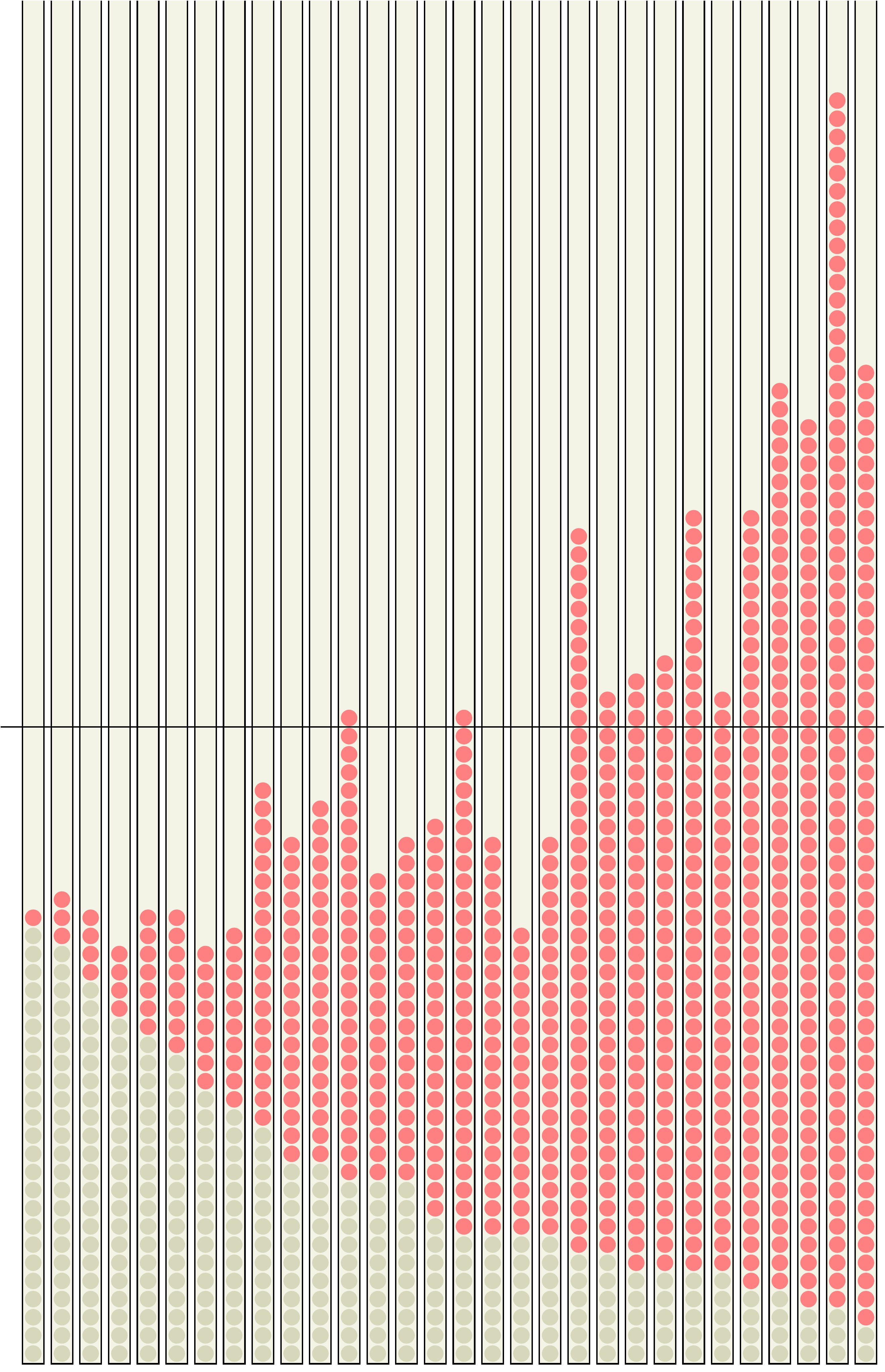} \\
    \TwoChoice
\end{center}
\end{minipage}\hspace{1cm}
\begin{minipage}[t]{0.3\textwidth}
\begin{center}
\includegraphics[scale=0.15]{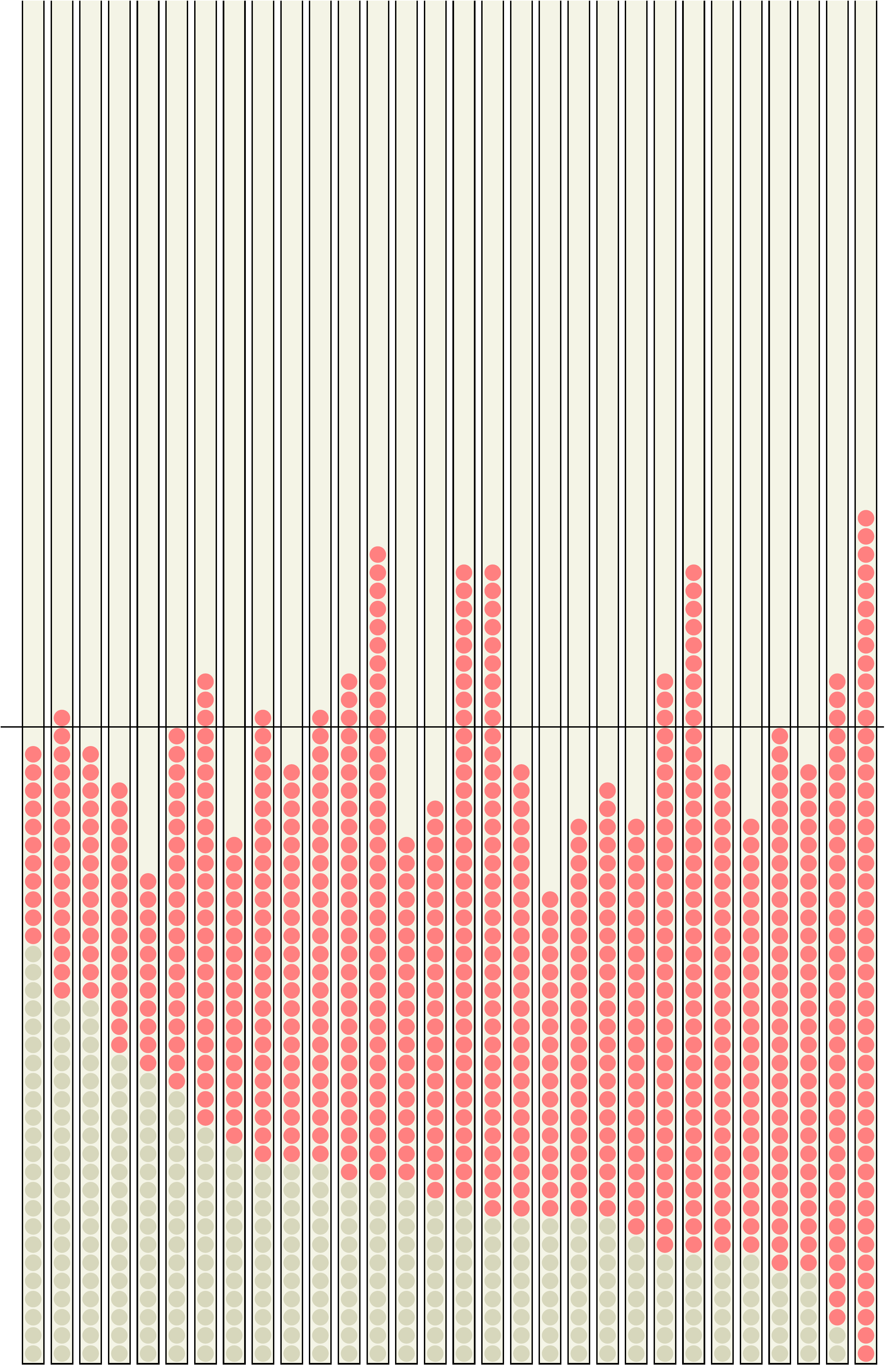} \\
    \OnePlusBeta-process
\end{center}
\end{minipage}
    \caption{The $b = 750$ balls of the latest batch shown in red allocated over the $n = 35$ bins (\textbf{left}) for \TwoChoice and (\textbf{right}) \OnePlusBeta with $\beta = 1/2$. Observe that \TwoChoice allocates more aggressively on the bins that are lightly loaded at the beginning of the batch, while \OnePlusBeta spreads the allocations more evenly.}    \label{fig:two_choie_vs_batch_visual}
\end{figure}

\paragraph{Our techniques.} Our techniques build on and refine those in~\cite{LS22Batched}, making use of the hyperbolic cosine potential function~\cite{PTW15} and variants. More specifically, a slightly weaker version of our tight upper bound is based on \cite[Theorem 3.1]{LS22Batched} and a refinement of \cite[Lemma 4.1]{LS22Batched}.  For our tight gap bound, our approach uses an interplay between two hyperbolic cosine potential functions to prove concentration and then an exponential potential with a larger smoothing parameter to deduce the refined gap.
A similar method was used in \cite[Section 5]{LS22Batched}, but one crucial novelty here is that we consider allocation processes whose probability allocation vector have a small $\ell_{\infty}$ distance from the uniform distribution. We believe that relating and comparing different allocation processes based on their $\ell_{\infty}$ distance (or other metrics) could be a promising avenue for future work. This can be also seen as a natural relaxation of the \emph{majorization technique}, which has been the dominant tool to relate different allocation processes \cite{PTW15,LS22Queries}.

\paragraph{Organization.} In \cref{sec:notation}, we introduce the basic notation for balanced allocations, and define the processes and settings that we will be working with. In particular, in \cref{sec:conditions} we define general conditions on the probability allocation vector used by the processes, under which our upper bounds on the gap apply. In \cref{sec:weak_gap}, we prove the $\Oh\big(\sqrt{b/n} \cdot \log n\big)$ bound on the gap for a family of processes in the weighted \Batched setting. In \cref{sec:strong_gap}, we perform a refined analysis and improve this bound to $\Oh\big(\sqrt{(b/n) \cdot \log n}\big)$.  
In \cref{sec:lower_bounds}, we show that this achieved gap is asymptotically optimal, and in~\cref{sec:experiments}, we present some empirical results on the gap of some specific processes.
Finally, in \cref{sec:conclusions}, we summarize the results and conclude with some open problems.

\begin{figure}[H]
    \centering
    \includegraphics{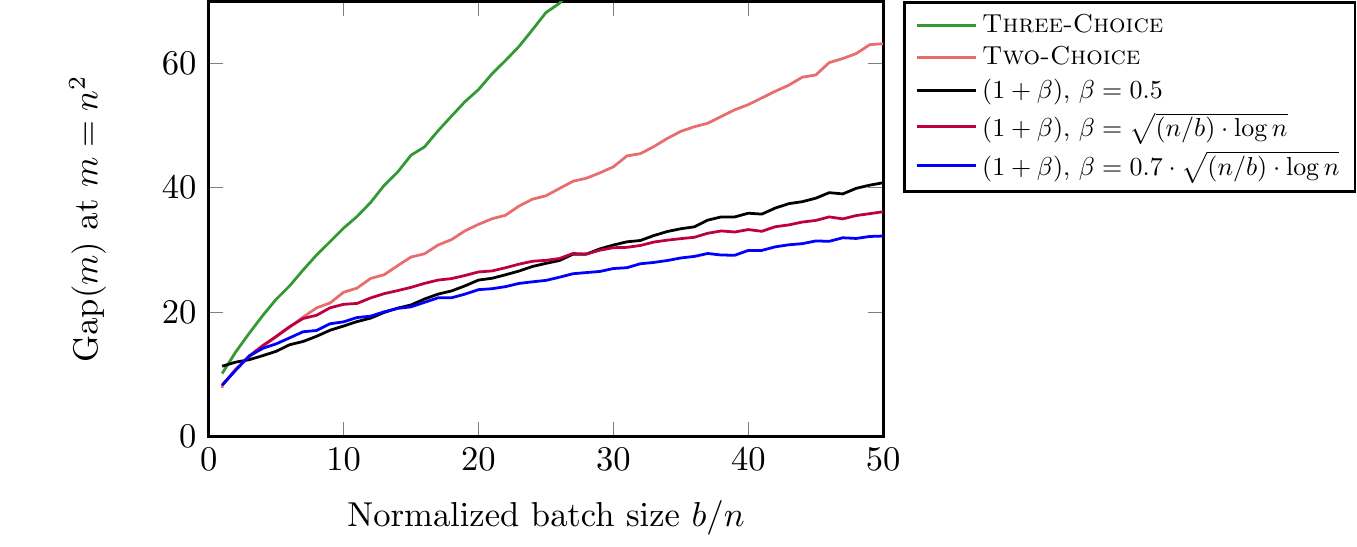}
    \definecolor{red2}{RGB}{232,109,109}
    \definecolor{green1}{RGB}{51,153,51}
    \caption{In the \Batched setting for large batch size $b$, the gaps achieved by the processes are ordered by their maximum entry in the probability allocation vector $p$: \textcolor{green1}{\ThreeChoice with $\max_{i \in [n]} p_i \approx \frac{3}{n}$}, \textcolor{red2}{\TwoChoice with $\max_{i \in [n]} p_i \approx \frac{2}{n}$}, \textcolor{black}{\OnePlusBeta with $\max_{i \in [n]} p_i \approx \frac{1 + \beta}{n}$} for $\beta = 0.5$, \textcolor{purple}{$\beta = \sqrt{(n/b) \cdot \log n}$} and \textcolor{blue}{$\beta = \sqrt{(n/b) \cdot \log n}$}. See \cref{fig:batched_unit_weights} for full details of the experiment.}
    \label{fig:batched_unit_weights_high_level}
\end{figure}

\begin{table}[H]
\centering
\resizebox{\textwidth}{!}{
\renewcommand{\arraystretch}{1.75}
 \begin{tabular}{ccccc}
        \textbf{Process} & \textbf{Gap in Sequential Setting } & \textbf{Gap in \Batched Setting} & \textbf{Batch Size}    \\ \hline
        \rowcolor{Gray} $\OneChoice$ & $\Theta \left(\sqrt{ (m/n) \cdot \log n} \right)$~\cite{RS98} &  $\Theta \left(\sqrt{ (m/n) \cdot \log n} \right)$~\cite{RS98} & $b \in \N$ \\ \hline 
        \rowcolor{Gray} & & $\Theta(\log \log n)$~\cite{LS22Noise} & $b \in \Theta(1)$ \\
        \rowcolor{Gray} & & $\Theta\left(\frac{\log n}{\log((4n/b) \log n)} \right)$~\cite{LS22Noise} & $b \in [n \cdot e^{-\log^{\Theta(1)} n}, n \log n]$ \\
        \rowcolor{Gray} \multirow{-3}{*}{$\TwoChoice$} & \multirow{-3}{*}{$\log_2 \log n + \Oh(1)$~\cite{ABKU99,BCSV06} } & $\Theta\left( b/n \right)$~\cite{LS22Batched} & $b \in [n \log n, n^3]$ \\ \hline
        \rowcolor{Gray} & & & \\ 
        \rowcolor{Gray} \multirow{-2}{*}{\makecell{$\OnePlusBeta$, \\ $\beta = \Theta(1)$}} & \multirow{-2}{*}{$\Theta(\log n)$~\cite{PTW15} } & \multirow{-2}{*}{$\Theta\left( b/n + \log n \right)$~\cite{LS22Batched}} & \multirow{-2}{*}{$b \in [n, n^3]$} \\ \hline
        \rowcolor{Greenish}  &  & $\Omega\left( \sqrt{ (b/n) \cdot \log n} \right)$~Thm~\ref{thm:lower} &  \\
        \rowcolor{Greenish} \multirow{-2}{*}{\makecell{\OnePlusBeta-process, \\ $\beta = \sqrt{(n/b) \cdot \log n}$}} & \multirow{-2}{*}{$\Theta(  (\log n)/\beta  ) $~\cite{PTW15}} & $\Oh\left( \sqrt{ (b/n) \cdot \log n} \right)$~Thm~\ref{thm:batching_strong_gap_bound} & \multirow{-2}{*}{$b \in [2n \log n,n^3]$} \\ \hline
        \rowcolor{Greenish} & & $\Omega \left( \sqrt{ (b/n ) \cdot \log n} \right)$~Thm~\ref{thm:lower} & \\
        \rowcolor{Greenish} \multirow{-2}{*}{\makecell{\OnePlusBeta-process, \\ $\beta = \sqrt{n/b}$}} & \multirow{-2}{*}{$\Theta(  (\log n)/\beta  ) $~\cite{PTW15}} & $\Oh\left( \sqrt{ b/n } \cdot \log n \right)$~Cor~\ref{cor:weak_bound} & \multirow{-2}{*}{$b \geq n \log n$}
             \\ \hline
        \end{tabular}
}
~\vspace{1em}~
\caption{Overview of the gap bounds in previous works (rows in \hlgray{~Gray~}) and the gap bounds derived in this work (rows in \hlgreenish{~Green~}). All gap bounds hold with probability at least $1 - o(1)$. Lower bounds hold for sufficiently large enough $m$. For the sake of simplicity, we focus on the setting with unit weights and only list results for \OnePlusBeta. Among all these processes, \OneChoice produces the worst gap in both settings, even though the gap does not change between the \Batched and sequential setting. For \TwoChoice, the gap becomes $b/n$ in the \Batched setting with $b = \Omega(n \log n)$, whereas for \OnePlusBeta the gap is improved to $\sqrt{(b/n) \cdot \log n}$ (for a suitable $\beta$).}
\label{tab:gap_summary}
\end{table}

\clearpage

\section{Notation, Processes and Settings} \label{sec:notation}

In this section, we introduce notation, processes and settings used throughout this work. %

\subsection{Basic Notation} \label{sec:basic_notation}

We consider the allocation of $m$ balls into $n$ bins, which are labeled $[n]:=\{1,2,\ldots,n\}$. For the moment, the $m$ balls are unweighted (or equivalently, all balls have weight $1$). For any step $t \geq 0$, $x^{t}$ is the $n$-dimensional \emph{load vector}, where $x_i^{t}$ is the number of balls allocated to bin $i$ in the first $t$ allocations. In particular, $x_i^{0}=0$ for every $i \in [n]$. Finally, the \emph{gap} is defined as
\[
 \Gap(t) = \max_{i \in [n]} x_i^{t} - \frac{t}{n}.
\]
It will also be convenient to sort the load vector $x$. To this end, let $\tilde{x}^t:=x^t-\frac{t}{n}$. Then, relabel the bins such that $y^{t}$ is a permutation of $\tilde{x}^t$ and $y_1^{t} \geq y_2^{t} \geq \cdots \geq y_n^{t}$. Note that $\sum_{i \in [n]} y_i^t=0$ and $\Gap(t)=y_1^t$. We call a bin $i \in [n]$ \emph{overloaded}, if $y_i^t \geq 0$ and \emph{underloaded} otherwise. 

A \textit{probability vector} $p \in \R^n$ is any vector satisfying $\sum_{i = 1}^n p_i = 1$ and $p_i \in [0,1]$ for $i \in [n]$. Following~\cite{PTW15}, many allocation processes can be described by a time-invariant \emph{probability allocation vector} $p^t$, which is the probability vector with $p_i^t$ being probability of allocating a ball to the $i$-th heaviest bin.

By $\mathfrak{F}^t$ we denote the filtration of the process until step $t$, which in particular reveals the load vector $x^t$.

\subsection{Processes} \label{sec:processes}

We start with a formal description of the \OneChoice process.

\begin{samepage}
\begin{framed}
\vspace{-.45em} \noindent
\underline{\OneChoice Process:} \\
\textsf{Iteration:} For each $t \geq 0$, sample one bin $i$, independently and uniformly at random. Then update:  
    \begin{equation*}
     x_{i}^{t+1} = x_{i}^{t} + 1.
 \end{equation*}\vspace{-1.5em}
\end{framed}
\end{samepage}
\noindent We continue with a formal description of the \TwoChoice process.
\begin{samepage}
\begin{framed}
\vspace{-.45em} \noindent
\underline{\TwoChoice Process:} \\
\textsf{Iteration:} For each $t \geq 0$, sample two bins $i_1$ and $i_2$, independently and uniformly at random. Let $i \in \{i_1, i_2 \}$ be such that $x_{i}^{t} = \min\{ x_{i_1}^t,x_{i_2}^t\}$, breaking ties randomly. Then update:  
    \begin{equation*}
     x_{i}^{t+1} = x_{i}^{t} + 1.
 \end{equation*}\vspace{-1.5em}
\end{framed}
\end{samepage}
It is immediate that the probability allocation vector of \TwoChoice is
\begin{equation*}
    p_{i} = \frac{2i-1}{n^2}, \qquad \mbox{ for all $i \in [n]$.}
\end{equation*}

Following~\cite{PTW15}, we recall the definition of the \OnePlusBeta-process which interpolates between \OneChoice and \TwoChoice:
\begin{samepage}
\begin{framed}
\vspace{-.45em} \noindent
\underline{($1+\beta$) Process:}\\
\textsf{Parameter:} A \textit{mixing factor} $\beta \in (0,1]$.\\
\textsf{Iteration:} For each $t \geq 0$, sample two bins $i_1$ and $i_2$, independently and uniformly at random. Let $i \in \{ i_1, i_2 \}$ be such that $x_{i}^{t} = \min\big\{ x_{i_1}^t,x_{i_2}^t \big\}$, breaking ties randomly. Then update:  
    \begin{equation*}
    \begin{cases}
     x_{i}^{t+1} = x_{i}^{t} + 1 & \mbox{with probability $\beta$}, \\
      x_{i_1}^{t+1} = x_{i_1}^{t} + 1 & \mbox{otherwise}.
   \end{cases}
 \end{equation*}\vspace{-1.em}
\end{framed}
\end{samepage}

In other words at each step, the \OnePlusBeta-process allocates the ball following the \TwoChoice rule with probability $\beta$, and otherwise allocates the ball following the \OneChoice rule. Therefore, its probability allocation vector is given by
\begin{equation*}
    p_{i} =
    (1-\beta) \cdot \frac{1}{n} + \beta \cdot \frac{2i-1}{n^2}, \qquad \mbox{ for all $i \in [n]$.}
\end{equation*}
Recall that in \cite{PTW15} (and \cite{LS22Batched}), it was shown that $\Gap(m) = \Oh\big(\frac{\log n}{\beta} \big)$ for any $m \geq n$ and $\beta \in (0, 1]$; so in particular, this gap (bound) does not grow with $m$.

The next process is another relaxation of \TwoChoice.
\begin{samepage}
\begin{framed}
\vspace{-.45em} \noindent
\underline{$\Quantile(\delta)$ Process:}\\
\textsf{Parameter:} A \textit{quantile} $\delta \in \{1/n, 2/n, \ldots, 1 \}$.\\
\textsf{Iteration:} For each $t \geq 0$, sample two bins $i_1$ and $i_2$, independently and uniformly at random. Then update:  
    \begin{equation*}
    \begin{cases}
     x_{i_2}^{t+1} = x_{i_2}^{t} + 1 & \mbox{if $i_1$ is among the $\delta n$ heaviest bins}, \\
     x_{i_1}^{t+1} = x_{i_1}^{t} + 1 & \mbox{otherwise}.
   \end{cases}
 \end{equation*}\vspace{-1.em}
\end{framed}
\end{samepage}
Note that the $\Quantile(\delta)$ processes can be implemented as a two-phase procedure: First probe the bin $i_1$ and place the ball there if $i_1$ is not among the $\delta n$ heaviest bins. Otherwise, take a second sample $i_2$ and place the ball there. Since we only need to know whether a bin's rank is above or below a value, the response by a bin can be encoded as a single bit (at the cost of knowing the rank of each bin). The probability allocation vector of $\Quantile(\delta)$ is given by:
\begin{equation*}
    p_{i} =
    \begin{cases}
     \frac{\delta}{n} & \mbox{ if $1 \leq i \leq \delta n$}, \\
     \frac{1+\delta}{n} & \mbox{ if $\delta n + 1 \leq i \leq n$}.
    \end{cases}
\end{equation*}

\subsection{Conditions on Probability Vectors}\label{sec:conditions}

In \cite{LS22Batched}, the weighted \Batched  setting was analyzed for probability allocation vectors satisfying the following two conditions. The first condition says that the process has a small $\eps/n$ bias to place away from overloaded and towards underloaded bins; and the second condition says that no bin has too high probability of being allocated.

\begin{itemize}\itemsep0pt
  \item \textbf{Condition $\mathcal{C}_1$}: There exist constant quantile\footnote{Here constant means that the quantile satisfies $\delta \in (\delta_1, \delta_2)$ for constant $\delta_1, \delta_2 \in (0, 1)$.} $\delta \in (0, 1)$ and (not necessarily constant) $\eps \in (0, 1)$, such that for any $1 \leq k \leq \delta n$,
    \[
    \sum_{i=1}^{k} p_{i} \leq (1 - \epsilon) \cdot \frac{k}{n},
    \]
    and similarly for any $\delta n +1 \leq k \leq n$,
    \[
     \sum_{i=k}^{n} p_i \geq \left(1 + \epsilon \cdot \frac{\delta}{1-\delta} \right) \cdot \frac{n-k+1}{n}.
    \]
 
  \item \textbf{Condition $\mathcal{C}_2$}: There exists a $C > 1$, such that $\max_{i \in [n]} p_i \leq \frac{C}{n}$.

\end{itemize}

In the same paper~\cite[Proposition 7.4]{LS22Batched} it was shown that any process with $\max_{i \in [n]} p_i \geq \frac{1+\eps}{n}$ for $\eps = \Omega(1)$ also has $\Gap(m) = \Omega(b/n)$ for any $b = \Omega(n \log n)$. Therefore, to improve on this asymptotic gao bound, we have to consider processes with $\max_{i\in [n]} p_i = \frac{1 + o(1)}{n}$. In our analysis in \cref{sec:weak_gap,sec:strong_gap} we will make use of the following condition based on the $\ell_\infty$-distance between the probability allocation vector $p$ and the uniform distribution (i.e., \OneChoice):%

\begin{itemize}
 \item \textbf{Condition $\mathcal{C}_3$}:  There exists a $C > 1$, such that
  \[
  \max_{i \in [n]} \left| p_i - \frac{1}{n}\right| \leq \frac{C - 1}{n}.
  \]
\end{itemize}

\noindent Note that this condition implies condition $\mathcal{C}_2$ for the same $C > 1$, but unlike $\mathcal{C}_2$ it imposes both an upper and a lower bound on the $p_i$'s. It is easy to verify that \OnePlusBeta-process satisfies all three conditions. 

\begin{lem} \label{lem:one_plus_beta_c123}
For any $\beta \in (0,1]$, the \OnePlusBeta-process satisfies condition $\mathcal{C}_1$ with $\delta=\frac{1}{4}$ and $\epsilon=\frac{\beta}{2}$, condition $\mathcal{C}_2$ with $C= 1+\beta$ and condition $\mathcal{C}_3$ with $C = 1 + \beta$.
\end{lem}
\begin{proof}
Recall that for the \OnePlusBeta-process, the probability allocation vector satisfies
\[
 p_i = (1-\beta) \cdot \frac{1}{n} + \beta \cdot \frac{2i-1}{n^2}, \quad \text{for all }i \in [n].
\]
We will first show that $\mathcal{C}_1$ holds with $\delta = 1/4$ and $\eps = \beta/2$. For any $1 \leq k \leq \delta n$, since $p$ is non-decreasing the prefix sums satisfy
\[
 \sum_{i = 1}^{k} p_{i} \leq p_k \cdot k \leq p_{\delta n} \cdot k \leq \left( (1-\beta)  + \beta \cdot (2\delta) \right) \cdot \frac{k}{n} = \left(1 - \frac{\beta}{2}\right) \cdot \frac{k}{n},
\]
Similarly, for any $\delta n + 1 \leq k \leq n$, the suffix sums satisfy 
\begin{align*}
\sum_{i = k}^n p_i 
 & \stackrel{(a)}{=} \frac{n - k + 1}{n} \cdot (1-\beta) + \frac{\beta}{n^2} \cdot (n^2 - (k-1)^2) \\
 & = \frac{n - k + 1}{n} \cdot (1-\beta) + \frac{\beta}{n^2} \cdot (n - k + 1) \cdot (n + k - 1) \\
 & =  \left(1 + \frac{\beta}{n} \cdot (k - 1) \right) \cdot \frac{n - k + 1}{n} \\
 & \stackrel{(b)}{\geq} (1 + \beta \cdot \delta ) \cdot \frac{n - k + 1}{n} \\
 & \stackrel{(c)}{\geq} \left(1 + \eps \cdot \frac{\delta}{1-\delta} \right) \cdot \frac{n - k + 1}{n},
\end{align*}
using in $(a)$ that $\sum_{i = 1}^u (2i - 1) = u^2$, in $(b)$ that $k \geq \delta n + 1$ and in $(c)$ that $\delta = 1/4$ and $\eps = \beta/2$.

Condition $\mathcal{C}_3$ for $C = 1 + \beta$ (and hence $\mathcal{C}_2$ as well) is verified as follows. As $p_i$ is increasing in $i \in [n]$,
\[
\max_{i \in [n]} \left\vert p_i - \frac{1}{n}\right\vert 
 = \max\left\{ \frac{1}{n} - p_1, p_n - \frac{1}{n} \right\} 
 = \frac{\beta}{n} - \frac{\beta}{n^2} 
 \leq \frac{\beta}{n}. \qedhere
\]
\end{proof}

Note that in contrast to \TwoChoice which satisfies $\mathcal{C}_3$ for $C = 2 - \frac{1}{n}$, by choosing $\beta$ small enough we can make the probability allocation vector arbitrarily close to uniform. 

We also note that for any process $\mathcal{P}$ satisfying condition $\mathcal{C}_3$ for some $C > 1$, we can define a process $\mathcal{P}'$ satisfying condition $\mathcal{C}_3$ for $C' \in (1, C)$ by mixing the probability allocation vector of $\mathcal{P}$ with that of \OneChoice with probability $\eta = \frac{C' - 1}{C - 1}$. 

For instance, the $\Quantile(1/2)$ process satisfies condition $\mathcal{C}_3$ for any $C = 1 + 1/2$ (since $\min_{i\in [n]} p_i = \frac{1}{2n}$ and $\max_{i \in [n]} p_i = \frac{3}{2n}$). Therefore, mixing $\Quantile(1/2)$ with \OneChoice with probability $\eta \in [0, 1]$, gives the following probability allocation vector satisfying condition $\mathcal{C}_3$ for $C = 1 + \eta/2$,
\[
p_i = \begin{cases}
\frac{1}{n} \cdot (1 - \eta) + \frac{1}{2n} \cdot \eta = \frac{1}{n} - \frac{\eta}{2n} & \text{if } i \leq \frac{1}{2}n, \\
\frac{1}{n} \cdot (1 - \eta) + \frac{3}{2n} \cdot \eta = \frac{1}{n} + \frac{\eta}{2n} & \text{otherwise}.
\end{cases}
\]

\begin{obs}
The process obtained by mixing $\Quantile(1/2)$ with \OneChoice satisfies condition $\mathcal{C}_1$ with $\delta = 1/2$ and $\eps = \eta/2$, condition $\mathcal{C}_2$ with $C = 1 + \eta/2$ and condition $\mathcal{C}_3$ with $C = 1 + \eta/2$.
\end{obs}

\subsection{Weighted and Batched Settings}

As in \cite{LS22Batched}, we now extend the definitions of \cref{sec:basic_notation} and \cref{sec:processes} to \emph{weighted balls} and later to the \emph{batched setting}. To this end, let
 $w^t \geq 0$ be the weight of the $t$-th ball to be allocated for $t \geq 1$. By $W^{t}$ we denote the total weights of all balls allocated after the first $t \geq 0$ allocations, so $W^t := \sum_{i=1}^n x_i^{t} = \sum_{s=1}^t w^s$. The normalized loads are $\tilde{x}_i^{t} := x_i^t - \frac{W^t}{n}$, and with $y_i^t$ being again the decreasingly sorted, normalized load vector, we have $\Gap(t)=y_1^t$. 

The weight of each ball will be drawn  independently from a fixed distribution $\mathcal{W}$ over $[0,\infty)$. Following~\cite{PTW15}, we assume that the distribution $\mathcal{W}$ satisfies:
\begin{itemize}
  \item $\ex{\mathcal{W}} = 1$.
  \item $\ex{e^{\zeta \mathcal{W}} } < \infty $ for some $\zeta > 0$.
\end{itemize}
Specific examples of distributions satisfying above conditions (after scaling) are the geometric, exponential, binomial and Poisson distributions.

In the analysis we will be using the following property (see also \cite{PTW15}) and refer to these distributions as $\FiniteMgf(\zeta)$ (or $\FiniteMgf(S)$):
\begin{lem}[{\cite[Lemma 2.4]{LS22Batched}}] \label{lem:bounded_weight_moment}
There exists $S := S(\zeta) \geq \max\{1, 1/\zeta\}$, such that for any  $\gamma \in (0, \min\{\zeta/2, 1\})$ and any $\kappa \in [-1,1]$,
\[
\Ex{e^{\gamma \cdot \kappa \cdot \mathcal{W}}} \leq 1 + \gamma \cdot \kappa + S \gamma^2 \cdot \kappa^2.
\]
\end{lem}

We will now describe the allocation of weighted balls into bins using a batch size of $ b \geq n$. For the sake of concreteness, let us first describe the \Batched setting if the allocation is done using \TwoChoice. For a given batch size consisting of $b$ consecutive balls, each ball of the batch performs the following. First, it samples two bins $i_1$ and $i_2$ independently and uniformly at random, and compares the load the two bins had at the beginning of the batch (let us denote the bin which has less load by $i_{\min}$). Secondly, a weight is sampled from the distribution $\mathcal{W}$. Then a weighted ball is added to bin $i_{\min}$. Recall that since the load information is only updated at the beginning of the batch, all allocations of the $b$ balls within the same batch can be performed in parallel.

In the following, we will use a more general framework, where the process of sampling (one or more) bins and then deciding where to allocate the ball to is described by a probability allocation vector $p$ over the $n$ bins (\cref{sec:basic_notation}). Also for the analysis, it will be convenient to focus on the normalized and sorted load vector $y$, which is why the definition below is based on $y$ rather than the actual load vector $x$.

\begin{samepage}
\begin{framed}
\vspace{-.45em} \noindent
\underline{\Batched Setting with Weights}\\
\textsf{Parameters:} Batch size $b \geq n$, probability allocation vector $p$, weight distribution $\mathcal{W}$.
\\
\textsf{Iteration:} For each $t = 0 \cdot b, 1 \cdot b, 2 \cdot b, \ldots$:
\begin{enumerate}\itemsep0pt
    \item Sample $b$ bins $i_1,i_2,\ldots,i_b$ from $[n]$ following $p$.
    \item Sample $b$ weights $w^{t+1},w^{t+2},\ldots,w^{t+b}$ from $\mathcal{W}$.
    \item Update for each bin $i \in [n]$, 
    \[
    z_{i}^{t+b}=y_{i}^{t} + \sum_{j=1}^b w^{t+j} \cdot \mathbf{1}_{i_j=i} - \frac{1}{n} \cdot \sum_{j=1}^b w^{t+j}.
    \]
    \item Let $y^{t+b}$ be the vector $z^{t+b}$, sorted decreasingly.
\end{enumerate}
\end{framed}
\end{samepage}

We also look at the version of the processes that perform random tie-breaking between bins of the same load. For $b = 1$, this makes no observable difference to the process, but for multiple steps, this effectively averages out the probability over (possibly) multiple bins that have the same load. This would, for instance, correspond to \TwoChoice, randomly deciding between the two bins if they have the same load. In particular, if $p$ is the original probability allocation vector, then the one with random tie-breaking is $\tilde{p}(y^t)$ (for $t$ being the beginning of the batch), where
\begin{equation} \label{eq:averaging_pi}
\tilde{p}_i(y^t) := \frac{1}{|\{ j \in [n] : y_j^t = y_i^t \}|} \cdot \sum_{j \in [n] : y_j^t = y_i^t} p_j, \quad \text{for all}i \in [n].
\end{equation}

\begin{samepage}
\begin{framed}
\vspace{-.45em} \noindent
\underline{\Batched Setting with Weights and Random Tie-Breaking}\\
\textsf{Parameters:} Batch size $b \geq n$, probability allocation vector $p$, weight distribution $\mathcal{W}$.
\\
\textsf{Iteration:} For each $t = 0 \cdot b, 1 \cdot b, 2 \cdot b, \ldots$:
\begin{enumerate}\itemsep0pt
    \item Let $\tilde{p} := \tilde{p}(y^t)$ be the probability allocation vector accounting for random tie-breaking.
    \item Sample $b$ bins $i_1,i_2,\ldots,i_b$ from $[n]$ following $\tilde{p}$.
    \item Sample $b$ weights $w^{t+1},w^{t+2},\ldots,w^{t+b}$ from $\mathcal{W}$.
    \item Update for each bin $i \in [n]$, 
    \[
    z_{i}^{t+b}=y_{i}^{t} + \sum_{j=1}^b w^{t+j} \cdot \mathbf{1}_{i_j=i} - \frac{1}{n} \cdot \sum_{j=1}^b w^{t+j}.
    \]
    \item Let $y^{t+b}$ be the vector $z^{t+b}$, sorted decreasingly.
\end{enumerate}
\end{framed}
\end{samepage}

\section{Warm-up: \texorpdfstring{$\Oh(\sqrt{b/n} \cdot \log n)$}{O(sqrt(b/n) log n} Gap} \label{sec:weak_gap}

In this section, we will refine the analysis of \cite[Section 4]{LS22Batched} to prove an $\Oh(\sqrt{b/n} \cdot \log n)$ bound on the gap for a family of processes. This will also be used as a starting point for the analysis in \cref{sec:strong_gap} to obtain the tighter bound. The main theorem that we prove is the following.

\begin{thm} \label{thm:herd_weak_gap_bound}
Consider any allocation process with probability allocation vector $p^t$ satisfying conditions $\mathcal{C}_1$ for constant $\delta \in (0, 1)$ and (not necessarily constant) $\eps \in (0,1)$ as well as condition $\mathcal{C}_3$ for some $C \in (1, 1.9)$, at every step $t \geq 0$. Further, consider the weighted \Batched setting with weights from a $\FiniteMgf(S)$ distribution with $S \geq 1$ and a batch size $b \geq \frac{2CS}{(C-1)^2} \cdot n$.
Then, there exists a constant $k := k(\delta) > 0$, such that for any step $m \geq 0$ being a multiple of $b$,
\[
\Pro{\max_{i \in [n]} |y_i^m| \leq k \cdot \frac{(C-1)^2}{\epsilon} \cdot \frac{b}{n} \cdot \log n } \geq 1 - n^{-2}.
\]
\end{thm}

Recall that by \cref{lem:one_plus_beta_c123}, the \OnePlusBeta-process satisfies condition $\mathcal{C}_1$ with $\eps = \frac{\beta}{2}$ and $\delta = \frac{1}{4}$, and conditions $\mathcal{C}_2$ and $\mathcal{C}_3$ with $C = 1 + \beta$.

In particular, by choosing $\beta = \Theta\big(\sqrt{n/b}\big)$ we get a process that is asymptotically better than \TwoChoice and which is within just a $\sqrt{\log n}$ multiplicative factor from the optimal gap bound proven for unit weights in \cref{sec:lower_bounds}.
\begin{cor}\label{cor:weak_bound}
Let $b \geq n \log n$ and consider the weighted \Batched setting with weights from a $\FiniteMgf(S)$ distribution with $S \in [1, b/4n]$. Then, there exists a constant $k > 0$ such that for the \OnePlusBeta-process with $\beta = \sqrt{4S \cdot \frac{n}{b}}$ and for any step $m \geq 0$ being a multiple of $b$,
\[
\Pro{\Gap(m) \leq k \cdot \sqrt{\frac{Sb}{n}} \cdot \log n} \geq 1 - n^{-2}.
\]
\end{cor}

The analysis is based on the \textit{hyperbolic cosine potential} which is defined for smoothing parameter $\gamma > 0$ as
\begin{align}
\Gamma^t := \Gamma^t(\gamma) := \Phi^t + \Psi^t := \sum_{i = 1}^n e^{\gamma y_i^t} + \sum_{i = 1}^n e^{-\gamma y_i^t}. \label{eq:hyperbolic}
\end{align}
We also decompose $\Gamma^t$ by defining
\[
 \Gamma_i^t := \Phi_i^t + \Psi_i^t = e^{\gamma y_i^t} + e^{-\gamma y_i^t}, \quad \text{for any bin $i \in [n]$}.
\]
Further, we use the following shorthands to denote the changes in the potentials over one step $\Delta\Phi_i^{t+1} := \Phi_i^{t+1} - \Phi_i^t$, $\Delta\Psi_i^{t+1} := \Psi_i^{t+1} - \Psi_i^{t}$ and $\Delta\Gamma_i^{t+1} := \Gamma_i^{t+1} - \Gamma_i^{t}$.

We will make use of the following drift theorem shown in \cite{LS22Batched}. Note that in statement of the theorem, \textit{rounds} could consist of multiple single-step allocations and in that case $p^t$ is not necessarily the probability allocation vector, but it could be a probability vector giving an estimate for the ``average number of balls'' allocated to a bin.

\newcommand{\MainHyperbolicCosineExpectation}{
Consider any allocation process $\mathcal{P}$ and a probability vector $p^t$ satisfying condition $\mathcal{C}_1$ for some constant $\delta \in (0, 1)$ and some $\eps \in (0, 1)$ at every round $t \geq 0$. Further assume that there exist $K > 0$, $\gamma \in \big(0, \min\big\{1, \frac{\eps\delta}{8K}\big\} \big]$ and $R > 0$, such that for any round $t \geq 0$, process $\mathcal{P}$ satisfies for potentials $\Phi := \Phi(\gamma)$ and $\Psi := \Psi(\gamma)$ that,
\[
\sum_{i = 1}^n \Ex{\left. \Delta\Phi_i^{t+1} \,\right|\, \mathfrak{F}^t} \leq \sum_{i = 1}^n \Phi_i^t \cdot \left(\left(p_i^t - \frac{1}{n}\right) \cdot R \cdot \gamma + K \cdot R \cdot \frac{\gamma^2}{n}\right),
\]
and
\[
\sum_{i = 1}^n \Ex{\left.\Delta\Psi_i^{t+1} \,\right|\, \mathfrak{F}^t} \leq  \sum_{i = 1}^n \Psi_i^t \cdot \left(\left(\frac{1}{n} - p_i^t\right) \cdot R \cdot \gamma + K \cdot R \cdot \frac{\gamma^2}{n}\right).
\]
Then, there exists a constant $c := c(\delta) > 0$, such that for $\Gamma := \Gamma(\gamma)$ and any round $t \geq 0$,
\[
\Ex{\left. \Delta\Gamma^{t+1} \,\right|\, \mathfrak{F}^t} \leq - \Gamma^t \cdot R \cdot \frac{\gamma\eps\delta}{8n} + R \cdot c\gamma\eps,
\]
and
\[
\Ex{\Gamma^t} \leq \frac{8c}{\delta} \cdot n.
\]}

\begin{thm}[{cf.~\cite[Theorem 3.1]{LS22Batched}}] \label{thm:hyperbolic_cosine_expectation}
\MainHyperbolicCosineExpectation
\end{thm}

Now we will show that any process satisfying condition $\mathcal{C}_3$, also satisfies the preconditions of \cref{thm:hyperbolic_cosine_expectation} for the expected change of the potential functions $\Phi$ and $\Psi$ over one batch.

\begin{lem} \label{lem:herd_batching_pot_changes}
Consider any allocation process with probability allocation vector $p^t$ satisfying condition $\mathcal{C}_3$ for some $C \in (1, 1.9)$ at every step $t \geq 0$. Further, consider the weighted \Batched setting with weights from a $\FiniteMgf(S)$ distribution with constant $S \geq 1$ and a batch size $b \geq \frac{2CS}{(C-1)^2} \cdot n$. Then for $\Phi := \Phi(\gamma)$ and $\Psi := \Psi(\gamma)$ with any smoothing parameter $\gamma \in (0, \frac{n}{2(C-1) \cdot b}]$ and any step $t \geq 0$ being a multiple of $b$,
\begin{align} \label{eq:phi_expected_precondition}
\Ex{\left. \Phi^{t+b} \,\right|\, \mathfrak{F}^t} \leq \sum_{i = 1}^n \Phi_i^t \cdot \left(1 + \Big(p_i^t -\frac{1}{n}\Big) \cdot b \cdot \gamma + \frac{5(C-1)^2b}{n} \cdot b \cdot \frac{\gamma^2}{n} \right),
\end{align}
and 
\begin{align} \label{eq:psi_expected_precondition}
\Ex{\left. \Psi^{t+b} \,\right|\, \mathfrak{F}^t} \leq \sum_{i = 1}^n \Psi_i^t \cdot \left(1 + \Big(\frac{1}{n} - p_i^t \Big) \cdot b \cdot \gamma + \frac{5(C-1)^2b}{n} \cdot b \cdot \frac{\gamma^2}{n} \right).
\end{align}
\end{lem}

The proof proceeds in a similar manner to \cite[Lemma 4.1]{LS22Batched}, but we bound the terms in \cref{eq:u_definition} and \cref{eq:tile_u_definition} more tightly using the new condition $\mathcal{C}_3$. Compared to the statement of \cite[Lemma 4.1]{LS22Batched}, the coefficients of the term $\frac{\gamma^2}{n}$ change from $5C^2 S^2 \frac{b^2}{n}$ to $5 (C-1)^2 \frac{b^2}{n}$. Note that $C$ is replaced by $C-1$, which makes a difference when $C = 1 + o(1)$, and that $S$ does not appear as we have assumed that $b \geq \frac{2CS}{(C-1)^2} \cdot n$.

\begin{proof}
Consider an arbitrary step $t \geq 0$ being a multiple of $b$ and for convenience let $p = p^t$. First note that the given assumptions $\gamma \leq \frac{n}{2(C-1) \cdot b}$ and $b \geq \frac{2CS}{(C-1)^2} \cdot n$ imply that 
\begin{align} \label{eq:alpha_second_bound}
\gamma \leq \frac{n}{2(C-1) \cdot b} \leq \frac{C-1}{4CS}.
\end{align}

Consider an arbitrary bin $i \in [n]$. Define the binary vector $Z \in \{0,1 \}^b$, where $Z_j$ indicates whether the $j$-th ball was allocated to bin $i$. The expected change for the overload potential $\Phi_i^t$ of the bin is given by,
\begin{align}
\Ex{\left. \Phi_i^{t+b} \,\right|\, \mathfrak{F}^t} 
& = \Phi_i^t \cdot \sum_{z \in \{0,1 \}^b} \Pro{Z = z} \cdot \Ex{\left. e^{\gamma \sum_{j = 1}^b \left(z_j w^{t+j} - \frac{w^{t+j}}{n} \right)} \, \right\vert \, \mathfrak{F}^t, Z = z} \notag .
 \end{align}
In the following, let us upper bound the factor of $\Phi_i^t$:
 \begin{align}
 & \sum_{z \in \{0,1 \}^b} \Pro{Z = z} \cdot \Ex{\left. e^{\gamma \sum_{j = 1}^b \left(z_j w^{t+j} - \frac{w^{t+j}}{n}\right)} \, \right\vert \, \mathfrak{F}^t, Z = z} \notag \\
 & \qquad \stackrel{(a)}{=} \!\!\! \sum_{z \in \{0,1 \}^b} \prod_{j = 1}^b (p_i)^{z_j}  (1 - p_i)^{1 - z_j}  (\ex{e^{\gamma W (1 - \frac{1}{n})}})^{z_j}  (\ex{e^{-\gamma W/n}})^{1- z_j} \notag \\
 & \qquad \stackrel{(b)}{\leq}\!\!\! \sum_{z \in \{0,1 \}^b} \prod_{j = 1}^b \left(p_i \cdot \left(1 + \gamma \cdot \left(1 - \frac{1}{n}\right) + S\gamma^2 \right)\right)^{z_j}  %
 \cdot \left((1 - p_i) \cdot \left(1 - \frac{\gamma}{n} + \frac{S\gamma^2}{n^2}\right) \right)^{1 - z_j} \notag \\
 & \qquad \stackrel{(c)}{=} \left( p_i \cdot \left(1 + \gamma \cdot \left(1 - \frac{1}{n}\right) + S\gamma^2 \right) + (1 - p_i) \cdot \left(1 - \frac{\gamma}{n} + \frac{S\gamma^2}{n^2}\right) \right)^b \notag \\
 & \qquad = \left( 1 + \gamma \cdot \left(p_i - \frac{1}{n}\right) + p_i \cdot S\gamma^2 + (1-p_i) \cdot \frac{S\gamma^2}{n^2} \right)^b \notag \\
 & \qquad \stackrel{(d)}{\leq} \left( 1 + \gamma \cdot \left(p_i - \frac{1}{n}\right) + 2 \cdot p_i \cdot S\gamma^2 \right)^b, \label{eq:phi_batched_i}
\end{align}
using in $(a)$ that the weights are independent given $\mathfrak{F}^t$, in $(b)$ \cref{lem:bounded_weight_moment} twice with $\kappa = 1 - \frac{1}{n}$ and with $\kappa = -\frac{1}{n}$ respectively (and that $(1 - 1/n)^2 \leq 1$), in $(c)$ the binomial theorem and in $(d)$ that $p_i \geq \frac{1}{n^2}$ by condition $\mathcal{C}_3$ for $C \in (1, 1.9)$. 
Let us define 
\begin{align} \label{eq:u_definition}
u_i := \left(p_i - \frac{1}{n}\right) \cdot \gamma + 2 \cdot p_i \cdot S\gamma^2.    
\end{align}
We will now show that $|u_i \cdot b| \leq 2(C-1) \cdot b \cdot \frac{\gamma}{n} \leq 1$, which holds indeed since
\begin{align}
|u_i \cdot b| &= \left\lvert \left(p_i -\frac{1}{n}\right) \cdot b \cdot \gamma + 2 \cdot p_i \cdot b \cdot S\gamma^2 \right\vert \notag \\
    & \leq \left\lvert \left(p_i -\frac{1}{n}\right) \cdot b \cdot \gamma\right\vert + 2 \cdot p_i \cdot b \cdot S\gamma^2 \notag \\
    & \stackrel{(a)}{\leq} \frac{C-1}{n} \cdot  b \cdot \gamma + 2 \cdot \frac{C}{n} \cdot b \cdot S \gamma^2 \notag \\
    & = ( C-1 + 2CS\gamma) \cdot b \cdot \frac{\gamma}{n} \notag \\
    & \stackrel{(b)}{\leq} 2(C-1) \cdot b \cdot \frac{\gamma}{n} \label{eq:yb_bounded_1} \\
    & \stackrel{(c)}{\leq} 1, \label{eq:yb_bounded_2}
\end{align}
using in $(a)$ that $\big|p_i - \frac{1}{n}\big| \leq \frac{C-1}{n}$ by condition $\mathcal{C}_3$, in $(b)$ that $\gamma \leq \frac{C-1}{2CS}$ by \cref{eq:alpha_second_bound} and in $(c)$ that $\gamma \leq \frac{n}{2(C-1) \cdot b}$. 

Then,
\begin{align*}
\Ex{\left. \Phi_i^{t+b} \,\right|\, \mathfrak{F}^t} 
 & \stackrel{(a)}{\leq} \Phi_i^t \cdot e^{u_i \cdot b} \\
 & \stackrel{(b)}{\leq} \Phi_i^t \cdot \left( 1 + u_i \cdot b + (u_i \cdot b)^2 \right) \\
 & \!\! \stackrel{(\ref{eq:u_definition})}{=} \Phi_i^t \cdot \left( 1 + \left(p_i - \frac{1}{n}\right) \cdot b \cdot \gamma + 2 \cdot p_i \cdot b \cdot S\gamma^2 + (u_i \cdot b)^2 \right) \\
 & \!\! \stackrel{(\ref{eq:yb_bounded_1})}{\leq} \Phi_i^t \cdot \left(1 + \left(p_i -\frac{1}{n}\right) \cdot b \cdot \gamma + 2 \cdot p_i \cdot b \cdot S\gamma^2 + \left(2(C-1) \cdot b \cdot \frac{\gamma}{n} \right)^2 \right) \\
 & \stackrel{(c)}{\leq} \Phi_i^t \cdot \left(1 + \left(p_i -\frac{1}{n}\right) \cdot b \cdot \gamma + \frac{5(C-1)^2b}{n} \cdot b \cdot \frac{\gamma^2}{n} \right),
 \end{align*}
 using in $(a)$ that $1 + v \leq e^v$ for any $v$, in $(b)$ that $e^v \leq 1 + v + v^2$ for $v \leq 1.75$ and \cref{eq:yb_bounded_2}, and in $(c)$ that $\frac{(C-1)^2b}{n} \cdot b \cdot \frac{\gamma^2}{n} \geq 2 \cdot \frac{C}{n} \cdot b \cdot S \gamma^2 \geq 2 \cdot p_i \cdot b \cdot S \gamma^2$, since $b \geq \frac{2CS}{(C-1)^2} \cdot n$.
 
Similarly, for the underloaded potential $\Psi^t$, for any bin $i \in [n]$,
\begin{align*}
\Ex{\left. \Psi_i^{t+b} \,\right|\, \mathfrak{F}^t} = \Psi_i^t \cdot \sum_{z \in \{0,1 \}^b} \Pro{Z = z} \cdot \Ex{\left. e^{-\gamma \sum_{j = 1}^b \left(z_j w^{t+j} - \frac{w^{t+j}}{n} \right)} \, \right\vert \, \mathfrak{F}^t, Z = z}. 
 \end{align*}
 As before, we will upper bound the factor of $\Psi_i^t$:
 \begin{align}
& \sum_{z \in \{0,1 \}^b} \Pro{Z = z} \cdot \Ex{\left. e^{-\gamma \sum_{j = 1}^b \left(z_j w^{t+j} - \frac{w^{t+j}}{n} \right)} \, \right\vert \, \mathfrak{F}^t, Z = z} \notag \\
 & \qquad \stackrel{(a)}{=} \!\!\! \sum_{z \in \{0,1 \}^b} \prod_{j = 1}^b (p_i)^{z_j}  (1 - p_i)^{1 - z_j}  (\ex{e^{-\gamma W \cdot (1 - \frac{1}{n})}})^{z_j}  (\ex{e^{\gamma W/n}})^{1- z_j} \notag \\
& \qquad \stackrel{(b)}{\leq}  \!\!\! \sum_{z \in \{0,1 \}^b} \prod_{j = 1}^b  \left(p_i \cdot \left(1 - \gamma \cdot \left(1 - \frac{1}{n}\right) + S\gamma^2 \right)\right)^{z_j} \notag %
\cdot \left((1 - p_i) \cdot \left(1 + \frac{\gamma}{n} + \frac{S\gamma^2}{n^2}\right) \right)^{1 - z_j} \notag \\
 & \qquad \stackrel{(c)}{=} \left( p_i \cdot \left(1 - \gamma \cdot \left(1 - \frac{1}{n}\right) + S\gamma^2 \right) + (1 - p_i) \cdot \left(1 + \frac{\gamma}{n} + \frac{S\gamma^2}{n^2}\right) \right)^b \notag \\
 & \qquad = \left( 1 + \left(\frac{1}{n}- p_i\right) \cdot \gamma + p_i \cdot S\gamma^2 + (1-p_i) \cdot \frac{S\gamma^2}{n^2} \right)^b \notag \\
 & \qquad \stackrel{(d)}{\leq} \left( 1 + \left( \frac{1}{n} - p_i\right) \cdot \gamma + 2 \cdot p_i \cdot S\gamma^2 \right)^b \label{eq:psi_batched_i},
\end{align}
using in $(a)$ that the weights are independent given $\mathfrak{F}^t$, in $(b)$ \cref{lem:bounded_weight_moment} twice with $\kappa = -\big(1 - \frac{1}{n}\big)$ and with $\kappa = \frac{1}{n}$ respectively, in $(c)$ the binomial theorem and in $(d)$ that $p_i \geq \frac{1}{n^2}$ by condition $\mathcal{C}_3$ for $C \in (1, 1.9)$.
Let us define
\begin{align} \label{eq:tile_u_definition}
\tilde{u}_i := \left(\frac{1}{n} - p_i\right) \cdot \gamma + 2 \cdot p_i \cdot S\gamma^2.
\end{align}
Similarly, to \cref{eq:yb_bounded_2}, we get that 
\begin{align}
|\tilde{u}_i \cdot b| 
 & \leq \left\lvert \left(\frac{1}{n} - p_i\right) \cdot b \cdot \gamma\right\vert + 2 \cdot p_i \cdot b \cdot S\gamma^2 \leq 2(C-1) \cdot b \cdot \frac{\gamma}{n} \label{eq:tilde_u_b_1} \\
 & \leq 1. \label{eq:tilde_u_b_2}
\end{align}
So,
\begin{align*}
 \Ex{\left. \Psi_i^{t+b} \,\right|\, \mathfrak{F}^t} 
 & \stackrel{(a)}{\leq} \Psi_i^t \cdot e^{\tilde{u}_i \cdot b} \\
 & \stackrel{(b)}{\leq} \Psi_i^t \cdot \left( 1 + \tilde{u}_i \cdot b + (\tilde{u}_i \cdot b)^2 \right) \\
 & \!\!\stackrel{(\ref{eq:tile_u_definition})}{=} \Psi_i^t \cdot \left(1 + \left(\frac{1}{n} - p_i\right) \cdot b \cdot \gamma + 2 \cdot p_i \cdot S \gamma^2 \cdot b + (\tilde{u}_i \cdot b)^2 \right) \\
 & \!\!\stackrel{(\ref{eq:tilde_u_b_1})}{\leq} \Psi_i^t \cdot \left(1 + \left(\frac{1}{n} - p_i\right) \cdot b \cdot \gamma + 2 \cdot p_i \cdot b \cdot \gamma^2 + \left(2(C-1) \cdot b \cdot \frac{\gamma}{n} \right)^2 \right) \\
 & \stackrel{(c)}{\leq} \Psi_i^t \cdot \left(1 + \left(\frac{1}{n} - p_i\right) \cdot b \cdot \gamma + \frac{5 (C-1)^2 b}{n} \cdot b \cdot \frac{\gamma^2}{n} \right),
\end{align*}
using in $(a)$ that $1 + v \leq e^v$ for any $v$, in $(b)$ that $e^v \leq 1 + v + v^2$ for $v \leq 1.75$ and \cref{eq:tilde_u_b_2}, and in $(c)$ that $\frac{(C-1)^2b}{n} \cdot b \cdot \frac{\gamma^2}{n} \geq 2 \cdot \frac{C}{n} \cdot b \cdot S \gamma^2 \geq 2 \cdot p_i \cdot b \cdot S \gamma^2$, since $b \geq \frac{2CS}{(C-1)^2} \cdot n$.
\end{proof}

Having verified the preconditions for \cref{thm:hyperbolic_cosine_expectation}, we are now ready to prove the bound on the gap for this family of processes.

\begin{rem}
The same upper bound in \cref{thm:herd_weak_gap_bound} also holds for processes with random tie breaking. The reason for this is that $(i)$ averaging probabilities in \cref{eq:averaging_pi} can only reduce the maximum entry (and increase the minimum) in the allocation vector $\tilde{p}^t$, i.e. $\max_{i \in [n]} \tilde{p}_i^t(x^t) \leq \max_{i \in [n]} p_i$, so it still satisfies condition $\mathcal{C}_3$ and $(ii)$ moving probability between bins $i, j$ with $x_i^t = x_j^t$ (and thus $\Phi_i^t = \Phi_j^t$ and $\Psi_i^t = \Psi_j^t$), implies that the aggregate upper bounds \cref{eq:phi_expected_precondition} and \cref{eq:psi_expected_precondition} in \cref{lem:herd_batching_pot_changes} remain the same.
\end{rem}

\begin{proof}[Proof of \cref{thm:herd_weak_gap_bound}]
Consider the \Batched setting at steps that are a multiple of $b$ and rounds consisting of $b$ consecutive allocations. By \cref{lem:herd_batching_pot_changes}, the preconditions of \cref{thm:hyperbolic_cosine_expectation} are satisfied for $K := 5 \cdot (C-1)^2 \cdot \frac{b}{n}$, $R := b$ and $\gamma := \frac{\eps\delta}{8K} = \frac{\eps\delta}{40 \cdot (C-1)^2 \cdot \frac{b}{n}} \leq \frac{n}{2(C-1) \cdot b}$, since $\eps \leq C - 1$ and also $\gamma \leq 1$ since $b \geq \frac{2CS}{(C-1)^2} \cdot n$, $C > 1$ and $S \geq 1$ (as in \cref{eq:alpha_second_bound}),
\[
\gamma \leq \frac{n}{2(C-1) \cdot b} \leq \frac{C-1}{4CS} \leq 1.
\]
Hence, by \cref{thm:hyperbolic_cosine_expectation}, there exists a constant $c := c(\delta) > 0$ such that for any step $m \geq 0$ which is a multiple of $b$,
\[
\Ex{\Gamma^m} \leq \frac{8c}{\delta} \cdot n.
\]
Therefore, by Markov's inequality
\[
\Pro{\Gamma^m \leq \frac{8c}{\delta} \cdot n^3} \geq 1 - n^{-2}.
\]
To conclude the claim, note that when $\big\{ \Gamma^m \leq \frac{8c}{\delta} \cdot n^3 \big\}$ holds, then also,
\[
\max_{i \in [n]} |y_i^m| \leq \frac{1}{\gamma} \cdot \left( \log \left( \frac{8c}{\delta}\right) + 3 \cdot \log n \right) \leq 4 \cdot \frac{\log n}{\alpha} \leq 4 \cdot \frac{8 \cdot 5 \cdot (C-1)^2}{\eps \delta} \cdot \frac{b}{n} \cdot \log n. \qedhere
\]
\end{proof}

\section{Tight Bound: \texorpdfstring{$\Oh(\sqrt{(b/n) \cdot \log n})$}{O(sqrt((b/n) log n))} Gap} \label{sec:strong_gap}

In this section, we will prove the stronger $\Oh\big(\sqrt{(b/n) \cdot \log n}\big)$ bound on the gap for a family of processes in the weighted \Batched setting (with $b \in [2n \log n, n^3]$). More specifically, these processes are a subset of the ones analyzed in \cref{sec:weak_gap} and include the \OnePlusBeta-process with $\beta = \sqrt{(n/b) \log n}$, as well as $\Quantile(1/2)$ mixed with \OneChoice. As we will show in \cref{sec:lower_bounds}, these processes achieve the asymptotically optimal bound.

\newcommand{\BatchingStrongGapBound}{
Consider the weighted \Batched setting with any $b \in [2n \log n, n^3]$ and weights from a $\FiniteMgf(S)$ distribution with constant $S \geq 1$. Further let $\eps = \sqrt{(n/b)  \cdot \log n}$.
Consider any process with probability allocation vector $p^t$ satisfying at every step $t \geq 0$, condition $\mathcal{C}_1$ for constant $\delta \in (0, 1)$ and $\eps$, as well as condition $\mathcal{C}_3$ for $C = 1 + \eps$. 
Then, there exists a constant $\kappa := \kappa(\delta, S) > 0$, such that for any step $m \geq 0$ being a multiple of $b$,
\[
\Pro{ \max_{i \in [n]} y_i^m \leq \kappa \cdot \sqrt{\frac{b}{n} \cdot \log n} } \geq 1 - n^{-2}.
\]
}

\begin{thm}\label{thm:batching_strong_gap_bound}
\BatchingStrongGapBound
\end{thm}

\noindent There are two key steps in the proof:

\textbf{Step 1:} Similarly to the analysis in \cite{LS22Queries}, we will use two instances of the hyperbolic cosine potential (defined in \cref{eq:hyperbolic}), in order to show that it is concentrated at $\Oh(n)$. More specifically, we will be using $\Gamma_1 := \Gamma_1(\gamma_1)$ with the smoothing parameter $\gamma_1 := \frac{\delta}{40S} \cdot \sqrt{n/(b \log n)}$ and $\Gamma_2 := \Gamma_2(\gamma_2)$ with $\gamma_2 := \frac{\gamma_1}{8 \cdot 30}$, i.e., with a smoothing parameter which is a large constant factor smaller than $\gamma_1$. So, in particular $\Gamma_2^t \leq \Gamma_1^t$ at any step $t \geq 0$. 

In the following lemma, proven in \cref{sec:batching_gamma_linear_whp}, we show that \Whp~$\Gamma_2 = \Oh(n)$ for any $\log^3 n$ consecutive batches. 

\newcommand{\BatchingGammaLinearWhp}{
Consider any process satisfying the conditions in \cref{thm:batching_strong_gap_bound}. Let $\tilde{c} := 2 \cdot \frac{8c}{\delta}$ where $c := c(\delta) > 0$ is the constant from \cref{thm:hyperbolic_cosine_expectation}. Then, for any step $t \geq 0$ being a multiple of $b$,
\[
\Pro{ \bigcap_{j \in [0, \log^3 n]} \left\lbrace \Gamma_2^{t + j \cdot b} \leq \tilde{c} n \right\rbrace } \geq 1 - n^{-3}.
\]
}

\begin{lem} \label{lem:batching_gamma_linear_whp}
\BatchingGammaLinearWhp
\end{lem}

The proof follows the interplay between the two hyperbolic cosine potentials, in that conditioning on $\Gamma_1^t = \poly(n)$ (which follows \Whp~by the analysis in \cref{sec:weak_gap}) implies that $\big|\Delta\Gamma_2^{t+1}\big| \leq n^{1/4} \cdot \sqrt{(n/b) \cdot \log n}$ (\cref{lem:batched_gamma_1_poly_implies}~$(ii)$). This in turn allows us to apply a bounded difference inequality to prove concentration for $\Gamma_2$. In contrast to \cite{LS22Queries} and \cite{LS22Noise}, here we need a slightly different concentration inequality \cref{lem:kutlin_3_3} (also used in \cite{LS22Batched}), as in a single batch the load of a bin may change by a large amount (with small probability).
The complete proof is given in \cref{sec:batching_gamma_linear_whp}.

\textbf{Step 2:} Consider an arbitrary step $s = t + j \cdot b$ where $\{ \Gamma_2^{s} \leq \tilde{c} n \}$ holds. Then, the number of bins $i$ with load $y_i^s$ at least $z := \frac{1}{\gamma_2} \cdot \log(\tilde{c} / \delta) = \Theta(\sqrt{(b/n) \cdot \log n})$ is at most $\tilde{c}n \cdot e^{-\gamma_2 z} = \delta n$. With this in mind, we define the following potential function for any step $t \geq 0$, which only takes into account bins that are overloaded by at least $z$ balls:
\[
\Lambda^t := \Lambda^t(\lambda, z) := \sum_{i : y_i^t \geq z} \Lambda_i^t := \sum_{i : y_i^t \geq z} e^{\lambda \cdot (y_i^t - z)},
\]
where $\lambda := \frac{\eps}{4CS} = \Theta(\sqrt{(n/b) \cdot \log n})$ and we define $\Lambda_i^t = 0$ for the rest of the bins $i$. This means that when $\{ \Gamma_2^{s} \leq \tilde{c} n \}$ holds, the probability of allocating to one of these bins is $p_i^{s} \leq \frac{1-\eps}{n}$, because of the condition $\mathcal{C}_1$. Hence, the potential drops in expectation over one batch (\cref{lem:lambda_drops}) and this means that \Whp~$\Lambda^m = \poly(n)$, which implies that $\Gap(m) = \Oh(z + \lambda^{-1} \cdot \log n) = \Oh(\sqrt{(b/n) \cdot \log n})$ gap.

\subsection{Step 1: Concentration of the \texorpdfstring{$\Gamma$}{Γ} Potential} \label{sec:batching_gamma_linear_whp}

Recall that in \cref{thm:batching_strong_gap_bound}, we considered the weighted \Batched setting with any $b \in [2n \log n, n^3]$ and weights sampled independently from a $\FiniteMgf(S)$ distribution with constant $S \geq 1$, for any allocation process with probability allocation vector $p^t$ satisfying condition $\mathcal{C}_1$ for constant $\delta \in (0, 1)$ and $\eps \in (0,1)$ as well as condition $\mathcal{C}_3$ for some $C > 1$, at every step $t \geq 0$. 

{\renewcommand{\thelem}{\ref{lem:batching_gamma_linear_whp}}
	\begin{lem}[Restated, page~\pageref{lem:batching_gamma_linear_whp}]
\BatchingGammaLinearWhp
	\end{lem} }
	\addtocounter{lem}{-1}

\begin{rem} \label{rem:load_vector_shape}
The number of bins $i \in [n]$ with normalized load value $y_i^t$ at least $z$ is \Whp~at most $\tilde{c} n \cdot e^{-\gamma_2 z}$, where $\gamma_2 = \Theta(\sqrt{(n/b) \cdot \log n})$.
\end{rem}

The proof of this lemma is similar to the proofs in \cite[Section 5]{LS22Batched} and \cite[Section 5]{LS22Queries}, in that we use the interplay between two instances of the hyperbolic cosine potential $\Gamma_1 := \Gamma_1 (\gamma_1)$ and $\Gamma_2 := \Gamma_2(\gamma_2)$ with smoothing parameter $\gamma_2$ being a large constant factor smaller than $\gamma_1$. More specifically, we will be working with $\gamma_1 := \frac{\delta}{40S} \cdot \sqrt{n/(b \log n)}$ and $\gamma_2 := \frac{\gamma_1}{8 \cdot 30}$.

The rest of this section is organized as follows. In \cref{sec:step_1_preliminaries}, we establish some basic properties for the potentials $\Gamma_1$ and $\Gamma_2$ and in \cref{sec:gamma_linear_whp_complete} we use these to show that \Whp~$\Gamma_2^t = \Oh(n)$ for at least $\log^3 n$ batches, and complete the proof of \cref{lem:batching_gamma_linear_whp}. Then, in \cref{sec:step_two}, we complete the proof of \cref{thm:batching_strong_gap_bound}.

\subsubsection{Preliminaries} \label{sec:step_1_preliminaries}

We define the following event, for any step $t \geq 0$\[
\mathcal{H}^t := \left\lbrace  w^t \leq \frac{15}{\zeta} \cdot \log n \right\rbrace,
\]
which means that the weight of the ball sampled in step $t$ is $\Oh(\log n)$ (since by assumption $\zeta > 0$ is constant). By a simple Chernoff bound and a union bound, we can deduce that this holds for a $\poly(n)$-long interval.

\begin{lem}[{cf.~\cite[Lemma 5.4]{LS22Batched}}] \label{lem:many_h_i}
Consider any $\FiniteMgf(\zeta)$ distribution $\mathcal{W}$ with constant $\zeta > 0$. Then, for any steps $t_0 \geq 0$ and $t_1 \in [t_0, t_0 + n^3 \log^3 n]$, we have that
\[
\Pro{\bigcap_{s \in [t_0, t_1]} \mathcal{H}^s} \geq 1 - n^{-10}
\]
\end{lem}

We will now show that when $\Gamma_1^t = \poly(n)$ and $\mathcal{H}^t$ holds, then $\Delta\Gamma_2^{t+1}$ is small.

\begin{lem} \label{lem:batched_gamma_1_poly_implies}
Consider any process satisfying the conditions in \cref{lem:batching_gamma_linear_whp} and any step $t \geq 0$, such that $\Gamma_1^{t} \leq 2\tilde{c} n^{26}$ and $\mathcal{H}^t$ holds. Then, we have that
\begin{align*}
(i) & \qquad \Gamma_2^t \leq n^{5/4}, \\
(ii) & \qquad |\Gamma_2^{t+1} - \Gamma_2^{t} | \leq n^{1/4} \cdot \sqrt{\frac{n}{b} \cdot \log n}.
\intertext{Further, let $\widehat{x}^t$ be the load vector obtained by moving any ball of the load vector $x^t$ to some other bin, then}
(iii) & \qquad \Gamma_1^t(\widehat{x}^t) \leq 2 \cdot \Gamma_1^t(x^t).
\end{align*}
\end{lem}
\begin{proof}
Recall that $\gamma_1 := \frac{\delta}{40S} \cdot \sqrt{n/(b \log n)}$ and $\gamma_2 := \frac{\gamma_1}{8 \cdot 30}$. Consider any step $t \geq 0$, such that $\Gamma_1^{t} \leq 2\tilde{c} n^{26}$ and $\mathcal{H}^t$ holds. We start by bounding the load of any bin $i \in [n]$,
\begin{align} \label{eq:batching_gamma2_load_bound}
\Gamma_1^{t} \leq 2\tilde{c} n^{26} 
 & \Rightarrow e^{\gamma_1\cdot y_i^{t}} + e^{-\gamma_1 \cdot y_i^{t}} \leq \tilde{c} n^{26}
 \Rightarrow 
y_i^t \leq \frac{27}{\gamma_1} \log n \, \wedge \,
 -y_i^t \leq \frac{27}{\gamma_1} \log n,
\end{align}
where in the second implication we used $\log (2\tilde{c}) + \frac{26}{\gamma_1} \log n \leq \frac{27}{\gamma_1} \log n$, for sufficiently large $n$.

\textit{First statement.} Using \cref{eq:batching_gamma2_load_bound}, we bound the contribution of any bin $i \in [n]$ to $\Gamma_2^t$ as follows,
\begin{equation} \label{eq:gamma_i_bound}
\Gamma_{2i}^t = e^{\gamma_2 y_i^t} + e^{-\gamma_2 y_i^t} \leq 2 \cdot e^{\gamma_2 \cdot \frac{27}{\gamma_1} \log n} \leq 2 n^{1/8},
\end{equation}
using that $\gamma_2 := \frac{\gamma_1}{8 \cdot 30}$. By aggregating, we get the first claim $\Gamma_1^t = \sum_{i = 1}^n \Gamma_{1i}^t \leq 2 \cdot n \cdot n^{1/8} \leq n^{5/4}$.

\textit{Second statement}. Let bin $j \in [n]$ be the bin where the $j$-th ball was allocated. We consider the following cases for the contribution of a bin $i$ to $\Gamma_{2i}^t$:

\textbf{Case 1 [$i = j$ and $y_j^t \geq 0$]:} Since $j \in [n]$ is overloaded, we have that
\begin{align*}
\left|\Delta\Gamma_{2j}^{t+1} \right| & \leq \Gamma_{2j}^t \cdot e^{\gamma_2 \cdot \frac{15}{\zeta} \cdot \log n}- \Gamma_{2j}^t 
  \stackrel{(a)}{\leq} \Gamma_{2j}^t \cdot \Big( 1 + \gamma_2 \cdot \frac{30}{\zeta} \cdot \log n \Big)- \Gamma_{2j}^t \\
& = \Gamma_{2j}^t \cdot \gamma_2 \cdot \frac{30}{\zeta} \cdot \log n 
 \stackrel{(b)}{\leq} n^{1/8} \cdot \sqrt{\frac{n}{b} \cdot \log n},
\end{align*}
using in $(a)$ the Taylor estimate $e^v \leq 1 + 2v$ (for $v \leq 1$) and that $\gamma_2 \cdot \frac{15}{\zeta} \cdot \log n \leq 1$, since $\gamma_2 \leq \frac{\delta}{60S} \cdot \sqrt{\frac{n}{b\log n}}$ and $S \geq \frac{1}{\zeta}$ and in $(b)$ that  \cref{eq:gamma_i_bound}, $\gamma_2 \leq \frac{\delta}{60S} \cdot \sqrt{\frac{n}{b\log n}}$ and $S \geq \frac{1}{\zeta}$.

\textbf{Case 2 [$i = j$ and $y_j^t < 0$]:} Similarly, if $j$ is underloaded, we have that
\begin{align*}
\left|\Delta\Gamma_{2j}^{t+1}\right| & \leq \Gamma_{2j}^t - \Gamma_{2j}^t \cdot e^{-\gamma_2 \cdot \frac{15}{\zeta} \cdot \log n} \leq \Gamma_{2j}^t - \Gamma_{2j}^t \cdot \Big( 1 - \gamma_2 \cdot \frac{30}{\zeta} \cdot \log n\Big) \\
 & = \Gamma_{2j}^t \cdot \gamma_2 \cdot \frac{30}{\zeta} \cdot \log n 
 \leq n^{1/8} \cdot \sqrt{\frac{n}{b} \cdot \log n}.
\end{align*}

\textbf{Case 3 [$i \neq j$ and $y_i^t \geq 0$]:} The contribution of the rest of the bins is due to the change in the average load. More specifically, for any overloaded bin $i \in [n] \setminus \{ j \}$,
\begin{align*}
\left|\Delta\Gamma_{2i}^{t+1}\right| 
 & \leq \Gamma_{2i}^t \cdot e^{ \gamma_2 \cdot \frac{15}{\zeta} \cdot \frac{\log n}{n}}- \Gamma_{2i}^t \leq \Gamma_{2i}^t \cdot \left( 1 + 2 \cdot \gamma_2 \cdot \frac{15}{\zeta} \cdot \frac{\log n}{n}\right)- \Gamma_{2i}^t \\
 & = \Gamma_{2i}^t \cdot \gamma_2 \cdot \frac{30}{\zeta} \cdot \frac{\log n}{n} 
 \leq \sqrt{\frac{\log n}{bn}} \cdot n^{1/8}.
\end{align*}

\textbf{Case 4 [$i \neq j$ and $y_i^t < 0$]:} Similarly, for any underloaded bin $i \in [n] \setminus \{ j \}$,
\begin{align*}
\left|\Delta\Gamma_{2i}^{t+1}\right| 
 & \leq \Gamma_{2i}^t - \Gamma_{2i}^t \cdot e^{-\gamma_2 \cdot \frac{15}{\zeta} \cdot \frac{\log n}{n}} \\
 & \leq \Gamma_{2i}^t - \Gamma_{2i}^t \cdot \left( 1 - 2 \cdot \gamma_2 \cdot \frac{15}{\zeta} \cdot \frac{\log n}{n}\right) \\
 & = \Gamma_{2i}^t \cdot \gamma_2 \cdot \frac{30}{\zeta} \cdot \frac{\log n}{n} 
 \leq \sqrt{\frac{\log n}{bn}} \cdot n^{1/8}.
\end{align*}
Hence, aggregating over all bins\begin{align*}
\left|\Delta\Gamma_2^{t+1}\right| 
 & \leq \left|\Delta\Gamma_{2j}^{t+1} \right| + \sum_{i \in [n] \setminus \{ j \}} \left|\Delta\Gamma_{2i}^{t+1} \right|   \leq n^{1/8} \cdot \sqrt{\frac{n}{b} \cdot \log n} + n \cdot \sqrt{\frac{\log n}{bn}} \cdot n^{1/8} 
  \leq n^{1/4} \cdot \sqrt{\frac{n}{b} \cdot \log n},
\end{align*}
for sufficiently large $n$.

\textit{Third statement.} Let $i, j \in [n]$ be the differing bins between $x^t$ and $\widehat{x}^t$. Then since $\mathcal{H}^t$ holds, it follows that $w^t \leq \frac{15}{\zeta} \cdot \log n$, so for bin $i$,
\[
\Gamma_{1i}^t(\widehat{x}^t) \leq e^{\gamma_1 w^t} \cdot \Gamma_{1i}^t(x^t) \leq 2 \cdot \Gamma_{1i}^t(x^t),
\]
since $\gamma_1 \leq \frac{1}{40 \cdot S \cdot \log n}$ and $S > 1/\zeta$. Similarly, for bin $j$.
\[
\Gamma_{1j}^t(\widehat{x}^t) \leq e^{\gamma_1 w^t} \cdot \Gamma_{1j}^t(x^t) \leq 2 \cdot \Gamma_{1j}^t(x^t),
\]
Hence, 
\[
\Gamma_1^t(\widehat{x}^t) = \sum_{k = 1}^n \Gamma_{1k}^t(\widehat{x}^t) \leq \sum_{k = 1}^n 2 \cdot \Gamma_{1k}^t(x^t) = 2 \cdot \Gamma_1^t(x^t). \qedhere
\]
\end{proof}

Next, we will show that $\ex{\Gamma_2} = \Oh(n)$ and that when $\Gamma_2$ is sufficiently large, it drops in expectation over the next batch.

\begin{lem}
\label{lem:large_gamma_exponential_drop}
Consider any process satisfying the conditions in \cref{lem:batching_gamma_linear_whp}. Then, there exists a constant $\tilde{c} := \tilde{c}(\delta)$ such that for any step $t\geq 0$ being a multiple of $b$, \[
(i) \quad \ex{\Gamma_1^t} \leq \frac{\tilde{c}}{2} \cdot n,
\quad \text{ and } \quad (ii) \quad \ex{\Gamma_2^t} \leq \frac{\tilde{c}}{2} \cdot n.
\]
Further, \[
(iii) \quad \Ex{\Gamma_2^{t+b} \,\, \Big\vert\,\, \mathfrak{F}^{t},\Gamma_2^t \geq \tilde{c} n} \leq 
\Gamma_2^{t} \cdot \Big(1-\frac{1}{\log n}\Big),
\]
and 
\[
(iv) \quad \Ex{\Gamma_2^{t+b} \,\,\Big\vert\,\, \mathfrak{F}^{t},\Gamma_2^t \leq \tilde{c} n} \leq 
\tilde{c} n - \frac{n}{\log^2 n}.
\]
\end{lem}
\begin{proof}
\textit{First/Second statement.} Recall that $C = 1+\eps$ and $\eps = \sqrt{\frac{n}{b} \cdot \log n}$. By \cref{lem:herd_batching_pot_changes} with $K := 5 \cdot (C - 1)^2 \cdot \frac{b}{n}$, $R := b$ and $\gamma_1 := \frac{\delta}{40S} \cdot \sqrt{\frac{n}{b \log n}}$, the preconditions for $\Phi := \Phi(\gamma_1)$ and $\Psi := \Psi(\gamma_1)$ in \cref{thm:hyperbolic_cosine_expectation} are satisfied. To apply \cref{thm:hyperbolic_cosine_expectation} we just need to verify $(i)$ that $\gamma_1 \leq \frac{\eps\delta}{8K} \leq 1$, which holds since
\[
\frac{\delta}{40S} \cdot \sqrt{\frac{n}{b \log n}}
 \leq \frac{\delta}{40} \cdot \sqrt{\frac{n}{b \log n}}
 = \frac{\eps\delta}{8K}
 \leq 1
\]
and $(ii)$ that $b \geq \frac{2CS}{(C-1)^2} \cdot n$, which follows since\[
\frac{2CS}{(C-1)^2} \cdot n = 2CS \cdot \frac{b}{\log n} \leq b.
\]
Hence, by \cref{thm:hyperbolic_cosine_expectation} we get the conclusion by setting $\tilde{c} := 16 c/\delta$, for some constant $c := c(\delta) > 0$.

Similarly for the potential $\Gamma_2 := \Gamma_2(\gamma_2)$ since $\gamma_2 \leq \gamma_1$. 

\textit{Third statement.} Furthermore, by \cref{lem:herd_batching_pot_changes} and \cref{thm:hyperbolic_cosine_expectation}, we also get that for any $t \geq 0$,
\begin{equation} \label{eq:tilde_gamma_drop}
\Ex{\left. \Gamma_2^{t+b} \,\,\right|\,\, \mathfrak{F}^t} \leq \Gamma_2^t \cdot \Big(1 - b \cdot \frac{\eps\delta}{8n} \cdot \gamma_2\Big) + b \cdot c\gamma_2\eps.
\end{equation}
We define the constant
\begin{align*}
\tilde{c}_1 
 & := \frac{1}{2} \cdot b \cdot \frac{\eps\delta}{8n} \cdot \gamma_2 
 = b \cdot \frac{\delta^2}{16 \cdot 8 \cdot 30 \cdot 40S \cdot n} \cdot \sqrt{\frac{n}{b \log n}} \cdot \sqrt{\frac{n \log n}{b}} \\
 & = \frac{\delta^2}{16 \cdot 8 \cdot 30 \cdot 40 \cdot S}.
\end{align*}
When $\big\{ \Gamma_2^{t} \geq \tilde{c} n\big \}$ holds, then \cref{eq:tilde_gamma_drop} yields,
\begin{align*}
\Ex{\Gamma_2^{t+b} \,\Big\vert\, \mathfrak{F}^t, \Gamma_2^{t} \geq \tilde{c} n} 
 & \leq \Gamma_2^{t} \cdot \Big(1 - 2 \cdot \tilde{c}_1 \Big) +  b \cdot c\gamma_2\eps \\
 & = \Gamma_2^{t} - \tilde{c}_1 \cdot \Gamma_2^{t} + \Big(b \cdot c\gamma_2\eps - \tilde{c}_1 \cdot \Gamma_2^{t}\Big)
 \\  
 & \leq \Gamma_2^{t} - \tilde{c}_1 \cdot \Gamma_2^{t} + \Big(b \cdot c\gamma_2\eps- \frac{1}{2} \cdot b \cdot \frac{\eps\delta}{8n} \cdot \gamma_2 \cdot \frac{16 c}{\delta} \cdot n \Big) \\ 
 & \leq \Gamma_2^{t} \cdot \Big(1-\frac{1}{\log n} \Big).
\end{align*}
\textit{Fourth statement.} Similarly, when $\Gamma_1^t < \tilde{c} n$, \cref{eq:tilde_gamma_drop} yields,
\begin{align*}
\Ex{\Gamma_2^{t+b} \,\Big\vert\, \mathfrak{F}^t, \Gamma_2^{t} < \tilde{c} n} 
 & \leq \tilde{c} n \cdot \Big(1 - 2 \cdot \tilde{c}_1\Big) + b \cdot c\gamma_2\eps \\
 & = \tilde{c} n - \tilde{c} \cdot \tilde{c}_1 \cdot n + \Big(b \cdot c\gamma_2\eps- \tilde{c} \cdot \tilde{c}_1 \cdot n \Big)
 \\ 
 &= \tilde{c} n - \frac{\tilde{c}}{\log n} \cdot n \leq \tilde{c} n - \frac{n}{\log^2 n}. \qedhere
\end{align*}
\end{proof}

In the next lemma, we show that \Whp~$\Gamma_1$ is $\poly(n)$ for every step in an interval of length $2 b \log^3 n$.

\begin{lem} \label{lem:gamma_continuous}
Let $\tilde{c} := 2 \cdot \frac{8c}{\delta}$ be the constant defined in \cref{lem:large_gamma_exponential_drop}. For any $b \in [2n \log n, n^3]$ and for any step $t \geq 0$ being a multiple of $b$,
\[
\Pro{ \bigcap_{s \in [t, t + 2b \log^ 3 n]} \left\{ \Gamma_1^{s} \leq \tilde{c} n^{26} \right\} } \geq 1 - n^{-10}.
\]
\end{lem}
\begin{proof}
We will start by bounding $\Gamma_1^s$ at steps $s$ being a multiple of $b$. Using \cref{lem:large_gamma_exponential_drop}~$(i)$, Markov's inequality and the union bound over $2 \log^3 n + 1$ steps, we have for any $t \geq 0$,
\begin{equation} \label{eq:base_union_bound}
\Pro{ \bigcap_{s \in [0, 2 \log^ 3 n]} \left\{ \Gamma_1^{t + s \cdot b} \leq \tilde{c} n^{12} \right\} } \geq 1 - (4\log^3 n) \cdot n^{-11}.
\end{equation}
Now, assuming that $\Gamma_1^{t + s \cdot b} \leq \tilde{c} n^{12}$, we will upper bound $\Gamma_1$ for the steps in between, i.e., for $\Gamma_1^{t + s \cdot b + r}$ for any $r \in [0, b)$. To this end, recalling that $\Gamma_{1i}^{t + s \cdot b + r} := \Phi_{1i}^{t + s \cdot b + r} + \Psi_{1i}^{t + s \cdot b + r}$, we will upper bound for each bin $i \in [n]$ the terms $\Phi_{1i}^{t + s \cdot b + r}$ and $\Psi_{1i}^{t + s \cdot b + r}$ separately. Proceeding using \cref{eq:phi_batched_i} in \cref{lem:herd_batching_pot_changes} (since $\gamma_1 \leq 1$ and $p$ satisfies condition $\mathcal{C}_3$),
\begin{align*}
\Ex{\left. \Phi_{1i}^{t + s \cdot b + r} \,\right|\, \mathfrak{F}^{t + s \cdot b}, \Phi_{1i}^{t + s \cdot b}}
 & \leq  \Phi_{1i}^{t + s \cdot b} \cdot \left( 1 + \left(p_i - \frac{1}{n}\right) \cdot \gamma_1 + 2 \cdot p_i \cdot S\gamma_1^2 \right)^r \\ 
 & \stackrel{(a)}{\leq} \Phi_{1i}^{t + s \cdot b} \cdot \left( 1 + \frac{C-1}{n} \cdot \gamma_1 + 2 \cdot \frac{C}{n} \cdot S\gamma_1^2 \right)^r \\ 
 & \stackrel{(b)}{\leq} \Phi_{1i}^{t + s \cdot b} \cdot \left( 1 + 2 \cdot \frac{C-1}{n} \cdot \gamma_1 \right)^r \\ 
 & \leq \Phi_{1i}^{t + s \cdot b} \cdot e^{2\gamma_1 (C-1) \cdot \frac{r}{n}} \\
 & \stackrel{(c)}{\leq} 2 \cdot \Phi_{1i}^{t + s \cdot b},
\end{align*}
using in $(a)$ that $p_i - \frac{1}{n} \leq \frac{C - 1}{n}$ by condition $\mathcal{C}_3$, $(b)$ that $\gamma_1 \leq \frac{C-1}{2CS}$ (as in \cref{eq:alpha_second_bound}) and in $(c)$~that $\gamma_1 \leq \frac{\delta}{40} \cdot \sqrt{\frac{n}{b \log n}}$ and $C - 1 = \sqrt{\frac{n \log n}{b}}$. Similarly, using \cref{eq:psi_batched_i} in \cref{lem:herd_batching_pot_changes},
\begin{align*}
\Ex{\left. \Psi_{1i}^{t + s \cdot b + r} \,\right|\, \mathfrak{F}^{t + s \cdot b}, \Psi_{1i}^{t + s \cdot b}}
 & \leq \Psi_{1i}^{t + s \cdot b} \cdot \left( 1 + \left( \frac{1}{n} - p_i\right) \cdot \gamma_1 + 2 \cdot p_i \cdot S\gamma_1^2 \right)^r \\ 
 & \stackrel{(a)}{\leq} \Psi_{1i}^{t + s \cdot b} \cdot \left( 1 + \frac{C-1}{n} \cdot \gamma_1 + 2 \cdot \frac{C}{\gamma_1} \cdot S\gamma_1^2 \right)^r \\ 
 & \stackrel{(b)}{\leq} \Psi_{1i}^{t + s \cdot b} \cdot \left( 1 + 2 \cdot \frac{C - 1}{n} \cdot \gamma_1 \right)^r \\ 
 & \leq \Psi_{1i}^{t + s \cdot b} \cdot e^{2\gamma (C-1) \cdot \frac{r}{n}} \\
 & \stackrel{(c)}{\leq} 2 \cdot \Psi_{1i}^{t + s \cdot b},
\end{align*}
using in $(a)$ that $\frac{1}{n} - p_i \leq \frac{C - 1}{n}$ by condition $\mathcal{C}_3$, $(b)$ that $\gamma_1 \leq \frac{C - 1}{2CS}$ and in $(c)$ that $\gamma_1 \leq \frac{\delta}{40} \cdot \sqrt{\frac{n}{b \log n}}$ and $C - 1 = \sqrt{\frac{n \log n}{b}}$. Hence, combining and aggregating over the bins,
\[
\Ex{\left. \Gamma_1^{t + s \cdot b + r} \, \right| \, \mathfrak{F}^{t + s \cdot b}, \Gamma_1^{t + s \cdot b}} \leq 2 \cdot \Gamma_1^{t + s \cdot b}.
\]
Applying Markov's inequality, for any $r \in [0, b)$,
\[
\Pro{\Gamma_1^{t + s \cdot b + r} \leq n^{14} \cdot \Gamma_1^{t + s \cdot b}} \geq 1 - 2 n^{-14}.
\]
Hence, by a union bound over the $2b \log^3 n \leq 2 n^3 \log^3 n$ possible steps (since $b \leq n^3$) for $s \in [0, 2\log^3 n]$ and $r \in [0, b)$,%
\begin{align}
\Pro{ \bigcap_{r \in [0, b]}\bigcap_{s\in [0, 2\log^3 n]} \left\{ \Gamma_1^{t + s \cdot b + r} \leq n^{14} \cdot \Gamma_1^{t + s \cdot b} \right\} } 
 \geq 1 - 2n^{-14} \cdot 2 b \log^3 n \geq 1 - \frac{1}{2} n^{-10}. \label{eq:double_intersection_lb}
\end{align}
Finally, taking the union bound of \cref{eq:base_union_bound} and \cref{eq:double_intersection_lb}, we conclude
\begin{align*}
\Pro{ \bigcap_{s \in [t, t + 2b \log^ 3 n]} \left\{ \Gamma_1^{s} \leq \tilde{c} n^{26} \right\} }
& \geq \Pr\Bigg[ \bigcap_{r \in [0, b]}\bigcap_{s\in [0, 2\log^3 n]} \left\{ \Gamma_1^{t + s \cdot b + r} \leq n^{14} \cdot \Gamma_1^{t + s \cdot b} \right\} \\
& \qquad \qquad \cap \bigcap_{s \in [0, 2 \log^ 3 n]} \left\{ \Gamma_1^{t + s \cdot b} \leq \tilde{c} n^{12} \right\} \Bigg]\\
 & \geq 1 - \frac{1}{2} n^{-10} - (4\log^3 n) \cdot n^{-11} \geq 1 - n^{-10}. \qedhere
\end{align*}
\end{proof}

We will now show that \Whp~there is a step every $b\log^3 n$ steps, such that the exponential potential $\Gamma_2$ becomes $\Oh(n)$. We call this the \textit{recovery} phase.

\begin{lem}[Recovery] \label{lem:batched_gamma_1_poly_n_implies_gamma_2_linear_whp}
Consider any process satisfying the conditions in \cref{lem:batching_gamma_linear_whp}. Let $\tilde{c} := 2 \cdot \frac{8c}{\delta}$ be the constant defined in \cref{lem:large_gamma_exponential_drop}. For any step $t \geq 0$ being a multiple of $b$,
\[
\Pro{\bigcup_{s \in [0, \log^3 n]} \left\lbrace \Gamma_2^{t + s \cdot b} \leq \tilde{c} n \right\rbrace} \geq 1 - 2 n^{-8}.
\]
\end{lem}

\begin{proof}
By \cref{lem:large_gamma_exponential_drop}~$(ii)$, using Markov's inequality at step $t$ being a multiple of $b$, we have
\begin{equation} \label{eq:basic_markov_tilde_gamma}
\Pro{\Gamma_2^{t} \leq \tilde{c} n^9} \geq 1 - n^{-8}.
\end{equation}
We will be assuming $\Gamma_2^{t} \leq \tilde{c} n^9$. By \cref{lem:large_gamma_exponential_drop}~$(iii)$, for any step $r \geq 0$, then \[
\Ex{\left. \Gamma_2^{r+1} \,\right|\, \mathfrak{F}^{r}, \Gamma_2^{r} > \tilde{c} n} \leq 
\Gamma_2^{r} \cdot \Big(1-\frac{1}{\log n}\Big).
\]
In order to prove that $\Gamma_2^{t + s \cdot b}$ is small for some $s \in [0, b \log^3 n]$, we define the ``killed'' potential function for any $r \in [0, \log^3 n]$,
\[
\widehat{\Gamma}_2^{t + r\cdot b} := \Gamma_2^{t + r \cdot b} \cdot \mathbf{1}_{\bigcap_{ s \in [0, r]} \{ \Gamma_2^{t + s \cdot b} > \tilde{c} n\} }.
\]
Note that $\widehat{\Gamma}_2^{t + r \cdot b} \leq \Gamma_2^{t + r \cdot b}$ and that $\big\{ \widehat{\Gamma}_2^{t + r \cdot b} = 0 \big\}$ implies that $\big\{ \widehat{\Gamma}_2^{t + r \cdot b + 1} = 0 \big\}$. Hence, the $\widehat{\Gamma}$ potential satisfies unconditionally the drop inequality of \cref{lem:large_gamma_exponential_drop}~$(iii)$, that is,
\[
\Ex{\left. \widehat{\Gamma}_2^{t + (r+1) \cdot b} \,\right|\, \mathfrak{F}^{t + r \cdot b}, \widehat{\Gamma}_2^{t + r \cdot b}} 
 \leq \widehat{\Gamma}_2^{t + r \cdot b} \cdot \Big(1-\frac{1}{\log n}\Big).
\]
Inductively applying this for $\log^3 n$ batches and using that $\widehat{\Gamma}_2^t \leq \Gamma_2^t \leq \tilde{c} n^9$, %
\[
\Ex{\left. \widehat{\Gamma}_2^{t + (\log^3 n) \cdot b} \,\, \right\vert \,\, \mathfrak{F}^{t}, \Gamma_2^{t} \leq \tilde{c} n^9} 
 \leq \Gamma_2^{t} \cdot \Big(1-\frac{1}{\log n}\Big)^{\log^3 n} 
  \leq \tilde{c} n^{9} \cdot e^{-\log^2 n} < n^{-7}. 
\]
So by Markov's inequality, \[
 \Pro{\left. \widehat{\Gamma}_2^{t + (\log^3 n) \cdot b} < n  \,\,\right|\,\, \Gamma_2^{t} \leq \tilde{c} n^9} \geq 1 - n^{-{8}}
\]
By combining with \cref{eq:basic_markov_tilde_gamma},
\begin{align*}
\Pro{\widehat{\Gamma}_2^{t + (\log^3 n) \cdot b} < n}
 & \geq \left( 1 - n^{-8}\right) \cdot \left( 1 - n^{-8}\right) \geq 1 - 2n^{-8}.
\end{align*}
Due to the definition of $\Gamma_2$, at any step $t \geq 0$, deterministically $\Gamma_2^t \geq 2n$. So,
we conclude that~w.p.~at least $1 - 2 n^{-8}$, we have that $\widehat{\Gamma}_2^{t + (\log^3 n) \cdot b} = 0$ or equivalently the event 
\[
\neg\bigcap_{ s \in [0, \log^3 n]} \left\{ \Gamma_2^{t + s \cdot b} > \tilde{c} n \right\},
\]
holds, which implies the conclusion.
\end{proof}

\subsubsection{Completing the Proof of Lemma~\ref{lem:batching_gamma_linear_whp}}\label{sec:gamma_linear_whp_complete}

We are now ready to prove \cref{lem:batching_gamma_linear_whp}, using a method of bounded differences with a bad event \cref{lem:kutlin_3_3} (\cite[Theorem 3.3]{K02}).

\begin{proof}[Proof of \cref{lem:batching_gamma_linear_whp}]

Our starting point is to apply
 \cref{lem:batched_gamma_1_poly_n_implies_gamma_2_linear_whp}, which proves that there is at least one step $t + \rho \cdot b \in [t - b\log^3 n, t]$ with $\rho \in [-\log^3 n, 0]$ such that the potential $\Gamma_2$ is small,
 \begin{align} \label{eq:starting_point}
\Pro{\bigcup_{\rho \in [- \log^3 n, 0]} \left\lbrace \Gamma_2^{t + \rho \cdot b} \leq \tilde{c} n \right\rbrace } &\geq 1 - 2 n^{-8}.
 \end{align}
Note that if $t < b \log^3 n$, then deterministically $\Gamma_2^0 = 2n \leq \tilde{c} n$ (which corresponds to $\rho = -t/b$).

We are now going to apply the concentration inequality \cref{lem:kutlin_3_3} to each of the batches starting at $t + \rho \cdot b, \ldots, t + (\log^3 n) \cdot b$ and show that the potential remains $ \leq \tilde{c} n$ at the last step of each batch. More specifically, we will show that for any $\tilde{r} \in [\rho, \log^3 n]$, for $r = t + b \cdot \tilde{r}$,
\[
\Pro{\left. \Gamma_2^{r+b} > \tilde{c} n \,\right|\, \mathfrak{F}^r, \Gamma_2^r \leq \tilde{c} n} \leq 3n^{-4}.
\]

We will show this by applying \cref{lem:kutlin_3_3} for all steps of the batch $[r, r + b]$. We define the good event \[
\mathcal{G}_r := \mathcal{G}_r^{r+b} := \bigcap_{s \in [r, r + b]} \left( \left\{ \Gamma_1^s \leq \tilde{c} n^{26} \right\} \cap \mathcal{H}^s \right),
\] 
and the bad event $\mathcal{B}_r := (\mathcal{G}_r)^c$. Using a union bound over \cref{lem:many_h_i} and \cref{lem:gamma_continuous},
\begin{align} \label{eq:bad_event_union_bound}
\Pro{ \bigcap_{s \in [t-b \log^3 n, t+ b \log^ 3 n]} \left( \left\{ \Gamma_1^{s} \leq \tilde{c} n^{26} \right\} \cap \mathcal{H}^s \right) } \geq 1 - 2n^{-10}.
\end{align}

Consider any $u \in [r, r + b]$. Further, we define the slightly weaker good event, $\tilde{\mathcal{G}}_r^u := \bigcap_{s \in [r, u]} \left( \left\{ \Gamma_1^s \leq 2\tilde{c} n^{26} \right\} \cap \mathcal{H}^s \right)$ and the ``killed'' potential, %
\[
\widehat{\Gamma}_r^u := \Gamma_2^u \cdot \mathbf{1}_{\tilde{\mathcal{G}}_r^u}.
\]
We will show that the sequence $\widehat{\Gamma}_r^r, \ldots , \widehat{\Gamma}_r^{r + b}$ is strongly difference-bounded by $(n^{5/4}, n^{1/4} \cdot \sqrt{(n/b) \cdot \log n}, 2n^{-10})$ (\cref{def:strongly_dif_bounded}).

Let $\omega \in [n]^{b}$ be an allocation vector encoding the allocations made in $[r, r + b]$. Let $\omega'$ be an allocating vector resulting from $\omega$ by changing one arbitrary allocation. It follows that,
\begin{align*}
\left| \widehat{\Gamma}_r^{r+b}(\omega) - \widehat{\Gamma}_r^{r+b}(\omega') \right| &\leq \max_{\tilde{\omega}} \widehat{\Gamma}_r^{r+b}(\tilde{\omega}) - \min_{\tilde{\omega}} \widehat{\Gamma}_r^{r+b}(\tilde{\omega}) \\
&\leq \max_{\tilde{\omega} \in \tilde{\mathcal{G}}_{r}^{r+b}} \Gamma_{2}^{r+b}(\tilde{\omega}) -0 
\\ & \leq n^{5/4},
\end{align*}
where in the last inequality we used \cref{lem:batched_gamma_1_poly_implies}~$(i)$ that for any $\tilde{\omega} \in \tilde{\mathcal{G}}_r^{r+b}$, we have $\widehat{\Gamma}_r^{r+b}(\tilde{\omega}) \leq \Gamma_{2}^{r+b}(\tilde{\omega}) \leq n^{5/4}$.

We will now derive a refined bound by additionally assuming that $\omega \in \mathcal{G}_r$. Then, for any $u \in [r, r + b]$,
\[
\Gamma_1^{u}(\omega') \leq 2 \cdot \Gamma_1^{u}(\omega) \leq 2 \tilde{c} n^{26},
\]
where the first inequality is by \cref{lem:batched_gamma_1_poly_implies}~$(iii)$. 
Hence $\omega' \in \tilde{\mathcal{G}}_r^{r+b}$, so $\mathbf{1}_{\tilde{\mathcal{G}}_r^{r+b}(\omega')} = 1$ and  $\widehat{\Gamma}_r^{r+b}(\omega') 
= \Gamma_{2}^{r+b}(\omega')$. Similarly, for $\omega \in \mathcal{G}_r \subseteq \tilde{\mathcal{G}}_r^{r+b}$, we have $\widehat{\Gamma}_r^{r+b}(\omega) 
= \Gamma_{2}^{r+b}(\omega)$ and by \cref{lem:batched_gamma_1_poly_implies}~$(ii)$,
\[ 
\left|\widehat{\Gamma}_r^{r+b}(\omega) - \widehat{\Gamma}_r^{r+b}(\omega')\right| = \left| \Gamma_2^{r+b}(\omega) - \Gamma_2^{r+b}(\omega') \right| \leq n^{1/4} \cdot \sqrt{\frac{n}{b} \cdot \log n}.
\]

Within a single batch all allocations are independent, so we apply \cref{lem:kutlin_3_3}, choosing $\gamma_k := \frac{1}{b}$ and $N := b$, which states that for any $T > 0$ and $\mu := \Ex{\left. \widehat{\Gamma}_r^{r + b} > \mu + T \,\right|\, \mathfrak{F}^r, \Gamma_2^r \leq \tilde{c} n}$,
\begin{align*}
& \Pro{\left. \widehat{\Gamma}_r^{r + b} > \mu + T \, \right|\, \mathfrak{F}^r, \Gamma_2^r \leq \tilde{c} n} \\
 & \quad  \leq \exp\left( - \frac{T^2}{2 \cdot \sum_{k = 1}^{b} (n^{1/4} \cdot \sqrt{\frac{n}{b} \cdot \log n} + n^{5/4} \cdot \frac{1}{b})^2 }\right) + 2 n^{-10} \cdot \sum_{k = 1}^b b. 
\end{align*}
By \cref{lem:large_gamma_exponential_drop}~$(iv)$, we have $\mu \leq \ex{\widehat{\Gamma}_r^{r+b} \mid \Gamma_2^r < \tilde{c} n}  \leq \ex{\Gamma_{2}^{r+b} \mid \Gamma_2^r < \tilde{c} n} \leq \tilde{c} n - n/\log^2 n$. Hence, for $T := n / \log^2 n$, since $2n \log n \leq b \leq n^3$, we have
\begin{align*}
\Pro{\left. \widehat{\Gamma}_r^{r+b} > \tilde{c} n \, \right|\, \mathfrak{F}^r, \Gamma_2^r \leq  \tilde{c} n}
& \leq \exp\left( - \frac{n^2/\log^4 n}{2 \cdot b \cdot (2 \cdot n^{1/4} \cdot \sqrt{\frac{n}{b} \cdot \log n})^2 }\right) + 2n^{-10} \cdot b^2 \\
& \leq \exp\left( - \frac{n^{1/2}}{8 \cdot \log^5 n} \right) + 2n^{-10} \cdot n^6 \leq 3n^{-4}. 
\end{align*}
Let $\mathcal{K}_{\rho}^{\tilde{r}} := \mathcal{G}_{\rho}^{t + \tilde{r} \cdot b} \cap \{ \Gamma_2^{t + \rho \cdot b} \leq \tilde{c} n \}$ for $\tilde{r} \in [\rho, \log^3 n]$. For any $\tilde{r} \geq \rho$, since $\mathcal{K}_\rho^{\tilde{r}+1} \subseteq \mathcal{K}_\rho^{\tilde{r}}$, we have%
\begin{align} \label{eq:k_killed}
\Pro{\left. \widehat{\Gamma}_r^{t + (\tilde{r}+1) \cdot b} \cdot \mathbf{1}_{\mathcal{K}_\rho^{\tilde{r}+1}} >  \tilde{c} n \,\right|\, \mathfrak{F}^r, \widehat{\Gamma}_r^{t + \tilde{r} \cdot b} \cdot \mathbf{1}_{\mathcal{K}_\rho^{\tilde{r}}} \leq  \tilde{c} n} \leq 3 n^{-4}.
\end{align}
By union bound of \cref{eq:starting_point} and \cref{eq:bad_event_union_bound}, 
\begin{align} 
\Pro{\bigcup_{\rho \in [-\log^3 n]} \mathcal{K}_{\rho}^{\log^3 n}} 
 & \geq \Pro{\mathcal{G}_{-\log^3 n}^{\log^3 n} \cap \bigcup_{\rho \in [- \log^3 n, 0]} \left\lbrace \Gamma_2^{t + \rho \cdot b} \leq \tilde{c} n \right\rbrace} \notag \\
 & \geq 1 - 2n^{-8} - 2n^{-10} \geq 1 - 3n^{-8}.\label{eq:exists_k_event}
\end{align}
Let \[
\mathcal{A} := \bigcap_{\tilde{r} \in [0, \log^3 n]} \left\lbrace \Gamma_2^{t + \tilde{r} \cdot b} \leq \tilde{c} n \right\rbrace,
\] and \[\mathcal{A}_{\rho} := \bigcap_{\tilde{r} \in [\rho, \log^3 n]} \left\lbrace \widehat{\Gamma}_r^{t + \tilde{r} \cdot b} \cdot \mathbf{1}_{\mathcal{K}_\rho^{\tilde{r}}} \leq \tilde{c} n \right\rbrace.\] Then, 
\begin{align*}
\Pro{\mathcal{A}_\rho \,\,\left|\,\, \Gamma_2^{t + \rho \cdot b} \leq \tilde{c} n \right.} 
  & \geq \prod_{\tilde{r} \in [\rho, \log^3 n - 1]} \mathbf{Pr} \left[ \bigcap_{\tilde{s} \in [\rho+1, \tilde{r}+1]} \left\lbrace \widehat{\Gamma}_r^{t + \tilde{s} \cdot b} \cdot \mathbf{1}_{\mathcal{K}_\rho^{\tilde{s}}} \leq \tilde{c} n \right\rbrace \right. \\
  &  \qquad \qquad  \left. \bigg\vert \, \bigcap_{\tilde{s} \in [\rho+1, \tilde{r} - 1]} \left\lbrace \widehat{\Gamma}_r^{t + \tilde{s} \cdot b} \cdot \mathbf{1}_{\mathcal{K}_\rho^{\tilde{s}}} \leq \tilde{c} n \right\rbrace, \widehat{\Gamma}_r^{t + \tilde{r} \cdot b} \cdot \mathbf{1}_{\mathcal{K}_\rho^{\tilde{s}}} \leq \tilde{c} n\right] \\
 & \geq \prod_{\tilde{r} \in [\rho, \log^3 n - 1]} \Pro{\left. \widehat{\Gamma}_r^{\tilde{r}+b} \cdot \mathbf{1}_{\mathcal{K}_\rho^{\tilde{r}+1}} > \tilde{c} n \,\right|\, \mathfrak{F}^{t + \tilde{r} \cdot b}, \widehat{\Gamma}_r^{t + \tilde{r} \cdot b} \cdot \mathbf{1}_{\mathcal{K}_\rho^{\tilde{r}}} \leq  \tilde{c} n} \\
 & \geq (1 - 3n^{-4})^{2\log^3 n} \geq 1 - 6n^{-4} \cdot \log^3 n,
\end{align*}
where in the last inequality we have used \cref{eq:k_killed} and the fact $\rho \geq -\log^3 n$. So,
\begin{align}
\Pro{\mathcal{A}_\rho} 
 & = \Pro{\mathcal{A}_\rho \,\left|\,\, \Gamma_2^{t + \rho \cdot b} \leq \tilde{c} n \right.} \cdot \Pro{\Gamma_2^{t + \rho \cdot b} \leq \tilde{c} n} + 1 \cdot \Pro{\neg \left\{ \Gamma_2^{t + \rho \cdot b} \leq \tilde{c} n \right\}} \notag \\
 & \geq 1 - 6n^{-4} \cdot \log^3 n. \label{eq:event_a_rho} 
\end{align}
Note that for any $\rho \in [-\log^3 n, 0]$, we have that $\mathcal{A}_\rho \cap \mathcal{K}_\rho^{\log^3 n} \subseteq \mathcal{A}$. Hence we conclude by the union bound of \cref{eq:exists_k_event} and \cref{eq:event_a_rho}, that
\begin{align*}
\Pro{\mathcal{A}} 
 & \geq \Pro{\bigcup_{\rho \in [-\log^3 n, 0]} \mathcal{K}_{\rho}^{\log^3 n} \cap \bigcap_{\rho \in [-\log^3 n, 0]} \mathcal{A}_\rho}  \geq 1  - 3n^{-8} - 6n^{-4} \cdot \log^6 n \geq 1- n^{-3}. \qedhere
\end{align*}
\end{proof}

\subsection{Step 2: Completing the Proof of Theorem~\ref{thm:batching_strong_gap_bound}}\label{sec:step_two}

We will now show that when $\Gamma_2^t = \Oh(n)$, the stronger potential function $\Lambda^t$ drops in expectation over the next batch. This will allow us to prove that $\Lambda^m = \poly(n)$ and deduce that \Whp~$\Gap(m) = \Oh\big(\sqrt{(b/n) \cdot \log n}\big)$.

\begin{lem} \label{lem:lambda_drops}
Consider any process satisfying the conditions in \cref{thm:batching_strong_gap_bound}. Let $\tilde{c} := 2 \cdot \frac{8c}{\delta}$ where $c := c(\delta) > 0$ is the constant from \cref{thm:hyperbolic_cosine_expectation}. For any step $t \geq 0$ being a multiple of $b$,
\[
\Ex{\left. \Lambda^{t+b} \,\right|\, \mathfrak{F}^t, \Gamma_2^t \leq \tilde{c} n } \leq \Lambda^t \cdot e^{-\frac{\lambda\eps}{2n} \cdot b} + n^2.
\]
\end{lem}
\begin{proof}
Consider an arbitrary step $t \geq 0$ being a multiple of $b$ and consider a labeling of the bins so that they are sorted by load. Assuming that $\{ \Gamma_2^t \leq \tilde{c} n \}$ holds, the number of bins with load $y_i^t \geq z$ is at most 
\[
\tilde{c} n \cdot e^{- \gamma_2 \cdot z} = \tilde{c} n \cdot e^{-\log(\tilde{c}/\delta)} = \delta n.
\]
For any bin $i \in [n]$ with $y_i^t \geq z$, we get as in \cref{eq:phi_batched_i} (using that $\lambda \leq 1$ and that $p$ satisfies $\mathcal{C}_3$ for $C \in (1, 1.9)$),%
\begin{align*}
\Ex{\left. \Lambda_i^{t + b} \,\,\right|\,\, \mathfrak{F}^t} 
 & \leq \Lambda_i^t \cdot \left( 1 + \Big(p_i^t - \frac{1}{n}\Big) \cdot \lambda + 2 \cdot p_i^t \cdot S\lambda^2 \right)^b.
\end{align*}
Since there are at most $\delta n$ such bins (i.e., $i \leq \delta n$), $p$ satisfies condition $\mathcal{C}_1$ and the normalized vector $y^t$ is sorted, by \cref{lem:quasilem2} the upper bound on $\Ex{\Lambda^{t+b} \mid \mathfrak{F}^t, \Gamma_2^t \leq \tilde{c} n}$ is maximized when $p_i^t = \frac{1 - \eps}{n}$, so
\begin{align*}
\sum_{i : y_i^t \geq z} \Ex{\left. \Lambda_i^{t + b} \,\,\right|\,\, \mathfrak{F}^t, \Gamma_2^t \leq \tilde{c} n} 
 & \stackrel{(a)}{\leq} \sum_{i : y_i^t \geq z} \Lambda_i^t \cdot \left( 1 - \frac{\lambda\eps}{n} + 2CS \cdot \frac{\lambda^2}{n} \right)^b \\
 &  \stackrel{(b)}{\leq} \sum_{i : y_i^t \geq z} \Lambda_i^t \cdot \left( 1 - \frac{\lambda\eps}{2n}\right)^b  
 \stackrel{(c)}{\leq} \sum_{i : y_i^t \geq z} \Lambda_i^t \cdot e^{-\frac{\lambda\eps}{2n} \cdot b},
\end{align*}
using in $(a)$ that $p_i^t \leq \frac{C}{n}$, in $(b)$ that $\lambda = \frac{\eps}{4CS}$ and in $(c)$ that $1 + v \leq e^v$ for any $v$. For the rest of the bins with $i > \delta n$,
\begin{align*}
\Ex{\left. \Lambda_i^{t + b} \,\right|\, \mathfrak{F}^t} 
 & \leq \Lambda_i^t \cdot \left( 1 + \left(p_i^t - \frac{1}{n}\right) \cdot \lambda + 2 \cdot p_i^t \cdot S\lambda^2 \right)^b \\
 & \stackrel{(a)}{\leq} \Lambda_i^t \cdot \left( 1 + \frac{C - 1}{n} \cdot \lambda + 2CS \cdot \frac{\lambda^2}{n} \right)^b \\
 & \stackrel{(b)}{\leq} \Lambda_i^t \cdot \left( 1 + 2 \cdot \frac{C-1}{n} \cdot \lambda\right)^b \\
 & \stackrel{(c)}{\leq} \left( 1 + 2 \cdot \frac{C-1}{n} \cdot \lambda\right)^b \\
 & \stackrel{(d)}{\leq} e^{2 \cdot \frac{C-1}{n} \cdot \lambda b} \stackrel{(e)}{\leq} n,
\end{align*}
using in $(a)$ that $p_i^t \leq \frac{C}{n}$, in $(b)$ that $2CS \cdot \frac{\lambda^2}{n} = \frac{\eps}{2n} \cdot \lambda \leq \frac{C-1}{2n} \cdot \lambda$ since $\lambda = \frac{\eps}{4CS}$ and $\eps = C - 1$, in $(c)$ that $\Lambda_i^t \leq 1$, in $(d)$ that $1 + v \leq e^v$ for any $v$ and in $(e)$ that $2 \cdot \frac{C-1}{n} \cdot \lambda b \leq \frac{\eps^2}{2CS} \cdot \frac{b}{n} \leq \log n$ (since $C > 1$ and $S \geq 1$).

Aggregating the contributions over all bins, 
\begin{align*}
\Ex{\left. \Lambda^{t+b} \,\right|\, \mathfrak{F}^t, \Gamma_2^t \leq \tilde{c} n } 
 & \leq \sum_{ i : y_i^t \geq z} \Lambda_i^t \cdot e^{-\frac{\lambda\eps}{2n} \cdot b} + \sum_{ i : y_i^t < z} n
 \leq \Lambda^t \cdot e^{-\frac{\lambda\eps}{2n} \cdot b} + n^2.\qedhere
\end{align*}
\end{proof}

Now we are ready to complete the proof of \cref{thm:batching_strong_gap_bound}.

\begin{proof}[Proof of \cref{thm:batching_strong_gap_bound}]
First consider the case when $m \geq b \cdot \log^3 n$. Let $t_0 = m - b \cdot \log^3 n$. Let $\mathcal{E}^t := \big\{ \Gamma_2^{t} \leq \tilde{c} n \big\}$. Then using \cref{lem:batching_gamma_linear_whp},
\begin{equation} \label{eq:eps_interval}
\Pro{ \bigcap_{j \in [0, \log^3 n]} \mathcal{E}^{t_0 + j \cdot b} } \geq 1 - n^{-3}.
\end{equation}
We define the killed potential $\tilde{\Lambda}$, with $\tilde{\Lambda}^{t_0} := \Lambda^{t_0}$ and for $j > 0$,
\[
\tilde{\Lambda}^{t_0 + j \cdot b} := \Lambda^{t_0 + j \cdot b} \cdot \mathbf{1}_{\cap_{s \in [0, j]} \mathcal{E}^{t_0 + s \cdot b}}.
\]
Since $\tilde{\Lambda}^t \leq \Lambda^t$, we have that by \cref{lem:lambda_drops} for $t = t_0 + j \cdot b$, we have that
\begin{align*}
\Ex{\tilde{\Lambda}^{t_0+(j+1)\cdot b} \,\,\left|\,\, \mathfrak{F}^{t_0+j \cdot b}, \Gamma_2^{t_0+j \cdot b} \leq \tilde{c} n \right.} \leq \tilde{\Lambda}^{t_0 + j \cdot b} \cdot e^{-\frac{\lambda\eps}{2n} \cdot b} + n^2.
\end{align*}
When $\mathcal{E}^{t_0 + j \cdot b}$ does not hold, then deterministically $\tilde{\Lambda}^{t_0 + (j+1) \cdot b} = \tilde{\Lambda}^{t_0 + j \cdot b} = 0$. Hence, we have the following unconditional drop inequality
\begin{align} \label{eq:lambda_drop}
\Ex{\tilde{\Lambda}^{t_0+(j+1)\cdot b} \,\,\left|\,\, \mathfrak{F}^{t_0+j \cdot b} \right.} \leq \tilde{\Lambda}^{t_0 + j \cdot b} \cdot e^{-\frac{\lambda\eps}{2n} \cdot b} + n^2.
\end{align}
Assuming $\mathcal{E}^{t_0}$ holds, we have %
\[
\max_{i \in [n]} y_i^{t_0} 
  \leq \frac{1}{\gamma_2} \cdot \left( \log \tilde{c} + \log n \right) 
  \leq \frac{2}{\gamma_2} \cdot \log n,
\]
for sufficiently large $n$. Recalling that $\gamma_2 = \Theta(\lambda \cdot \log n)$, there exists a constant $\kappa_1 > 0$ such that
\[
\tilde{\Lambda}^{t_0} \leq n \cdot e^{\lambda \cdot y_1^{t_0}} \leq e^{\kappa_1 \log^2 n}.
\]
Applying \cref{lem:geometric_arithmetic} to \cref{eq:lambda_drop} with $a := e^{-\frac{\lambda\eps}{2n} \cdot b}$ and $b := n^2$ for $\log^3 n$ steps, %
\begin{align}
\Ex{\tilde{\Lambda}^{m} \,\,\left|\,\, \mathfrak{F}^{t_0}, \tilde{\Lambda}^{t_0} \leq e^{\kappa_1 \log^2 n} \right.} 
& \leq e^{\kappa_1 \log^2 n} \cdot a^{\log^3 n} + \frac{b}{1 - a}   \stackrel{(a)}{\leq} 1 + 1.5 \cdot b \leq  2n^2  \label{eq:poly_n_expectation}.
\end{align}
using in $(a)$ that $\frac{\lambda\eps}{2n} \cdot b = \Omega(\log n)$, since $\lambda = \frac{\eps}{4CS}$ and $\eps = \sqrt{(n/b) \cdot \log n}$.

By Markov's inequality, we have
\begin{align*}
\Pro{\tilde{\Lambda}^{m} \leq 2n^{5} \,\left|\, \mathfrak{F}^{t_0}, \tilde{\Lambda}^{t_0} \leq e^{\kappa_1 \log^2 n} \right.} \geq 1 - n^{-3}.
\end{align*}
Hence, by \cref{eq:eps_interval},
\begin{align*} %
\Pro{\tilde{\Lambda}^{m} \leq 2n^{5}} 
 & = \Pro{\left. \tilde{\Lambda}^{m} \leq 2n^{5} \,\right|\, \mathcal{E}^{t_0}} \cdot \Pro{\mathcal{E}^{t_0}} 
 \geq \left(1 - n^{-3}\right) \cdot \left(1 - n^{-3}\right) \geq 1 - 2n^{-3}.
\end{align*}
Combining with \cref{eq:eps_interval}, we have
\begin{align*}
\Pro{\Lambda^{m} \leq 2n^{5}} 
 &\geq \Pro{\left\lbrace\tilde{\Lambda}^{m} \leq 2n^{5} \right\rbrace \cap \bigcap_{j \in [0, \log^3 n]} \mathcal{E}^{t_0 + j \cdot b}} \geq 1 - 2n^{-3} - n^{-3} \geq 1 - n^{-2}.
\end{align*}
Finally, $\{ \Lambda^{m} \leq 2 n^{5} \}$ implies that
\[
\max_{i \in [n]} y_i^{m} \leq z + \frac{\log 2}{\lambda} + \frac{5\log n}{\lambda} = \Oh\left(\sqrt{(b/n) \cdot \log n} \right),
\]
since $\lambda = \frac{\eps}{4CS} = \Theta(\sqrt{(n \log n)/b})$, so the claim follows.

For the case when $m < b \cdot \log^3 n$, it deterministically holds that $\tilde{\Lambda}^{t_0} \leq n$, which is a stronger starting point in \cref{eq:poly_n_expectation} to prove that $\ex{\Lambda^m} \leq 2n^{5}$, which in turn implies the gap bound.
\end{proof}

\section{Lower Bounds on the Gap} \label{sec:lower_bounds}

In this section, we prove two lower bounds of $\Omega( \sqrt{ (b/n) \cdot \log n})$ on the gap. Both lower bounds hold even in the unit weights case.
\begin{obs}\label{obs:simple_lower}
Consider the \Batched setting with any $b \geq n \log n$, and assume all balls have unit weights. Then, for any process which uses the same probability allocation vector within each batch with random tie breaking, 
\[
\Pro{\Gap(b) \geq \frac{1}{10} \cdot \sqrt{(b/n) \cdot \log n}} \geq 1 - n^{-2}.
\]
\end{obs}
\begin{proof}
Any such process behaves exactly like \OneChoice in the first batch and so the lower bound follows from that of \OneChoice for $b$ balls into $n$ bins (cf.~\cite{RS98} and~\cite[Lemma A.2]{LS23RBB}).
\end{proof}

The next lower bound is more involved. This bound also applies to processes which are allowed to adjust the probability allocation vector from one batch to another arbitrarily;  e.g., the probability for a heavily underloaded bin might be set close to (or even equal to) $1$, and similarly, the probability for a heavily overloaded bin might be set close to (or equal to) $0$. Additionally, the lower bound below applies to any two consecutive batches, and not only to the end of the first batch as in \cref{obs:simple_lower}.

\begin{thm}\label{thm:lower}
Consider the \Batched setting with any $b = \Omega(n \log n)$ in the unit weights case. Furthermore, consider an allocation process which may adaptively change the probability allocation vector for each batch. Then there is a constant $\kappa > 0$ such that for any allocation process (which may adaptively change the probability for each batch) it holds that for every $t \geq 0$ being a multiple of $b$,
\[
 \Pro{ \max \left\{ \Gap(t), \Gap(t+b) \right\} \geq \kappa \cdot \sqrt{ (b/n) \cdot \log n}  } \geq 1/2. 
\]
\end{thm}
\begin{proof}
In the proof, we shall prove a slightly stronger statement:
\[
 \Pro{ \max \left\{ \Gap(t), \Gap(t+b) \right\} \geq \kappa \cdot \sqrt{ (b/n) \cdot \log n} ~~\Bigg|~~ \mathfrak{F}^t } \geq 1/2. 
\]
That is, there is no load configuration and no probability allocation vector (depending on $\mathfrak{F}^t$) such that the gap is small, both before and at the end of an arbitrary batch. 

For notational convenience, we will prove this statement by assuming that $t=0$, and $x^{0}$ is an arbitrary load vector satisfying $\sum_{i \in [n]} x_i^0=0$ (in other words, we shift time backwards by $t$ steps) and $p = p^0$ is the probability allocation vector used by the process.
Consider one arbitrary bin $j \in [n]$. Then,
\begin{align*}
 \Ex{ x_j^b - \frac{b}{n} + z_j} = b \cdot p_j - \frac{b}{n} + z_j =: \varphi_j %
\end{align*}

For a sufficiently large constant $C > 0$,  let us now assume $\max_{j \in [n]} z_j \leq C/2 \cdot \sqrt{(b/n) \cdot \log n}$; clearly, if this is not the case, we already have a large gap already before the next batch. 

Next consider a bin $j \in [n]$ with
\[
 p_j \leq \frac{1}{n} + \frac{1}{b} \cdot \left( -10 C \cdot \sqrt{(b/n) \cdot \log n} \right).
\]
We will now apply a Chernoff bound (\cref{lem:chernoff}) for $x_j^b \sim \mathsf{Bin}(b,p_j)$,
with $\delta := C \cdot \sqrt{(n/b) \cdot \log n}$, $\mu := b \cdot p_j$ and $\mu_H := \frac{b}{n} - 10 C \cdot \sqrt{(b/n) \cdot \log n} \geq \mu$ to get that
\[
  \Pro{ x_j^b \geq \frac{b}{n} - C \cdot \sqrt{(b/n) \cdot \log n}}
    \leq \Pro{ x_j^b \geq \mu_H \cdot (1 + \delta) }
    \leq e^{-\delta^2\mu_H/3} = e^{-\frac{C^2}{3} \cdot \log n} \leq n^{-4},
\]
using that $C \geq 4$.
If $\left\{ x_j^b \leq \frac{b}{n} - C \cdot \sqrt{(b/n) \cdot \log n} \right\}$ occurs, then
\[
x_j^b - \frac{b}{n} + z_j \leq -C \cdot \sqrt{(b/n) \cdot \log n} + z_j  \leq -C/ 2 \cdot \sqrt{(b/n) \cdot \log n} \leq 0,
\]
and thus bin $j$ will not contribute to the gap at step $b$.

Hence in the remainder of the proof, we would like to assume that for all bins $j \in [n]$,
\[
 p_j \geq \frac{1}{n} + \frac{1}{b} \cdot \left( -10 C \cdot \sqrt{(b/n) \cdot \log n} \right) =: p_{\operatorname{low}}
\]
Note $p_{\operatorname{low}} \leq 1/n$.
Consider now a transformation of the probability vector $(p_i)_{i \in [n]}$ into $(\tilde{p}_i)_{i \in [n]}$, where $\tilde{p}$ satisfies for all $j \in [n]$,$
p_{\operatorname{low}} \leq \tilde{p}_j \leq \max \left\{  p_{\operatorname{low}},  p_j \right\}.
$
In other words, in $\tilde{p}$ we only increase probabilities of bins $j \in [n]$, for which $p_j < p_{\operatorname{low}}$. Let us define $\mathcal{J} := \left\{ j \in [n] \colon p_j < p_{\operatorname{low}} \right\}$. For $b \geq (20C^2) n \log n$, this implies $\tilde{p}_j \geq \frac{1}{2n}$ for all $j \in [n]$.

Further,
let $(x_i^b)_{i \in [n]}$ be a load vector where the locations of the next $b$ balls are sampled according to $p$, and $(\tilde{x}_i^b)_{i \in [n]}$ be a load vector where these locations are sampled according to $\tilde{p}$. Clearly, there is a coupling so that for every $j \in [n] \setminus \mathcal{J}$, $x_j^b \geq \tilde{x}_j^b$ (since $p_j^b \geq \tilde{p}_j^b$). Further, for any $j \in \mathcal{J}$, by a union bound,
\[
 \Pro{ \max_{j \in \mathcal{J}} x_j^b > 0 } \leq | \mathcal{J} | \cdot n^{-4} \leq n^{-3}.
\]
Hence it follows that, for any threshold $T > 0$,
\begin{align*}
 \Pro{ \max_{j \in [n]} x_j^b \geq T } &\geq
 \Pro{ \max_{j \in [n]} \tilde{x}_j^b \geq T} - \Pro{ \max_{j \in \mathcal{J}} x_j^b \geq T}
  \geq  \Pro{ \max_{j \in [n]} \tilde{x}_j^b \geq T} - n^{-3}.
 \end{align*}

Therefore, in the remainder of the proof, we will lower bound $\Pro{ \max_{j \in [n] \setminus \mathcal{J}} \tilde{x}_j^b \geq T}$ for a suitable value of $T=\Omega( \sqrt{b / n \cdot \log n})$. We will also use the definition
\[
 \tilde{\varphi}_j := b \cdot \tilde{p}_j - \frac{b}{n} + z_j.
\]
Finally, we define $\xi=0.1$ as a (sufficiently) small constant.

\textbf{Case 1:} We have at least $n-n^{\xi}$ bins for which $\tilde{\varphi}_j \leq - C \sqrt{b/n}$. Since $\sum_{i \in [n]} \tilde{\varphi}_i = 0$, this implies that there must be at least one bin with $j \in [n]$ with
$
 \tilde{\varphi}_j \geq \frac{n-n^{\xi}}{n^{\xi}} \cdot C \cdot \sqrt{b/n} \geq 1/2 \cdot n^{1-\xi} \cdot \sqrt{b/n}.
$
Further, using that the median of a $\mathsf{Bin}(N, q)$ r.v.~is either $\lfloor Nq \rfloor$ or $\lceil Nq \rceil$, then
\[
\Pro{ x_j^b \geq \Ex{ x_j^b } - \frac{1}{4} n^{1-\xi} } \geq \Pro{ x_j^b \geq \left\lfloor \Ex{ x_j^b }\right\rfloor }
\geq 1/2, \] it follows that with probability at least $1/2$ we will have a large gap.

\textbf{Case 2:} We have at least $n^{\xi}$ bins with $\tilde{\varphi}_j \geq - C \sqrt{b/n}$; call this set $\mathcal{B}$. We further know that, due to the definition of $\tilde{p}$, we have for all bins $j \in [n]$ that $\tilde{p}_j \geq \frac{1}{2n}$. Hence, we set $T:= b \cdot \tilde{p}_j + \kappa \cdot \sqrt{ b \cdot \tilde{p}_j \cdot \log n}$, and applying \cref{lem:binomial} yields for any bin $j \in \mathcal{B}$,
$
 \Pro{ \tilde{x}_j^{b} \geq T } \geq n^{-\xi/2}.
$
Since $\tilde{p}_j \geq \frac{1}{2n}$, $|\mathcal{B}| \geq n^{\xi}$, the claim follows.\end{proof}

\section{Experimental Results} \label{sec:experiments}

In this section, we complement our theoretical analysis with some experimental results for the \Batched setting. In \cref{fig:beta_vs_batch_size}, we plot the gap of the \OnePlusBeta-process for various batch sizes and different values of $\beta \in (0, 1]$ (\TwoChoice corresponding to $\beta = 1$). The plot strongly suggests the existence of an optimal $\beta$, which seems to increase as the batch size $b$ grows.

\begin{figure}[H]
    \centering
    \vspace{-0.2cm}
    \includegraphics[scale=0.85]{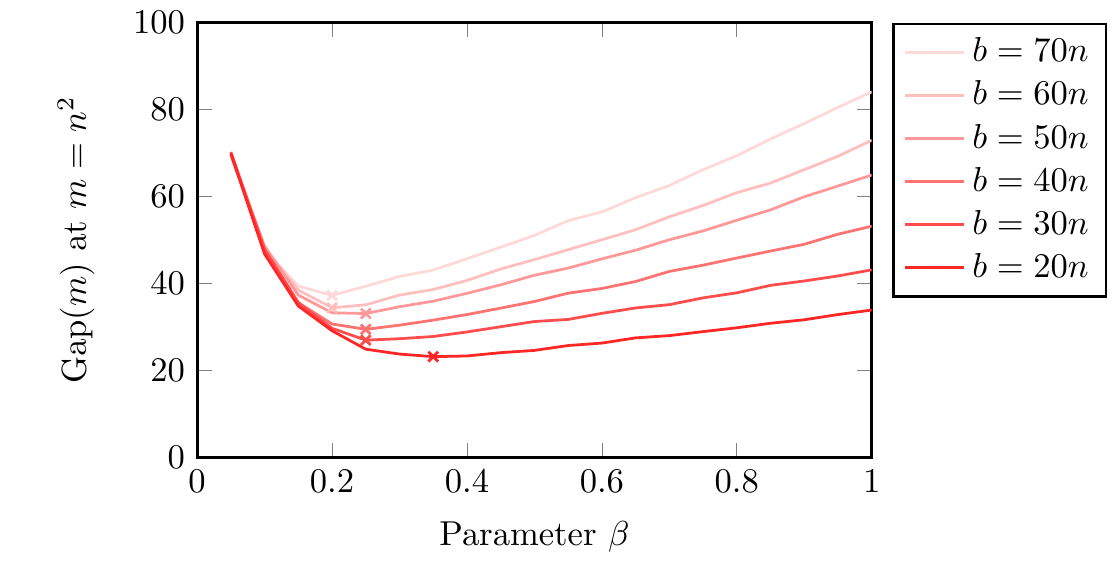}
    \vspace{-0.4cm}
    \caption{Average gap for the \OnePlusBeta-process in the \Batched setting with unit weights for $n = 1.000$ bins, $m = n^2$ balls, for various batch sizes and parameter values $\beta \in (0, 1]$ (averaged over $25$ runs).}
    \label{fig:beta_vs_batch_size}
\end{figure}

In \cref{fig:eta_vs_batch_size}, we present the corresponding empirical results of \cref{fig:beta_vs_batch_size} for  
the \Quantile process (mixed with \OneChoice). As with the \OnePlusBeta-process, the optimal mixing factor $\eta$ tends to increase as the batch size grows. The \Quantile with the optimized mixing factor seems to perform slightly worse than the optimized \OnePlusBeta-process.

\begin{figure}[H]
    \centering
    \vspace{-0.2cm}
    \includegraphics[scale=0.85]{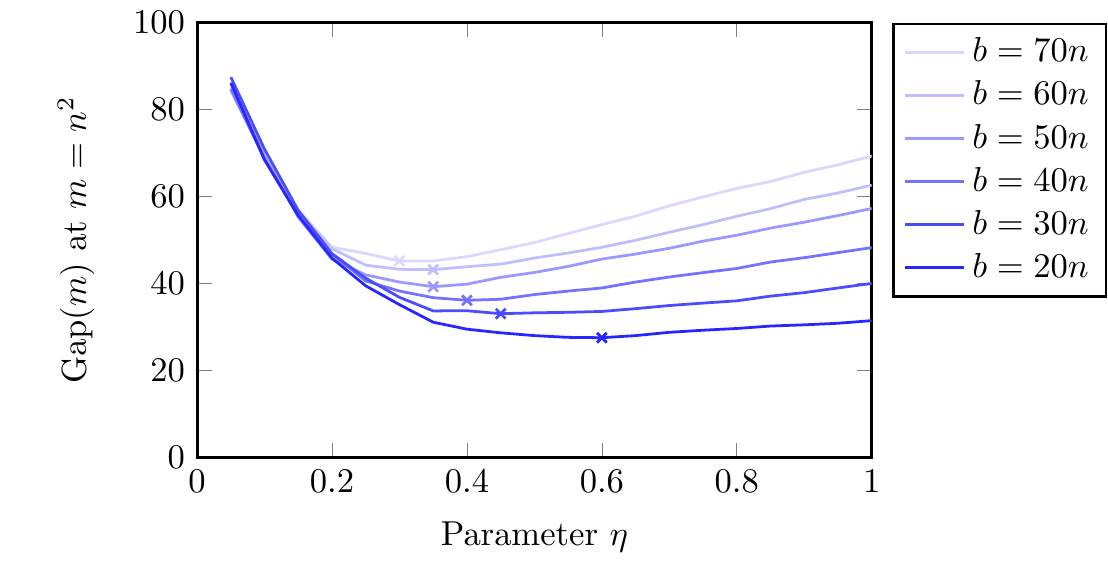}
    \vspace{-0.4cm}
    \caption{Average gap for the \Quantile process (mixed with \OneChoice with probability $\eta \in (0, 1]$) in the \Batched setting with unit weights for $n = 1.000$ bins, $m = n^2$ balls, for various batch sizes and parameter values $\eta \in [0, 1]$ (averaged over $25$ runs).}
    \label{fig:eta_vs_batch_size}
\end{figure}

 In \cref{fig:batched_unit_weights}, we plot the gap of \TwoChoice, \ThreeChoice and \OnePlusBeta versus the batch size. For small values of $b$, the gap of \TwoChoice and \ThreeChoice is small, but soon grows rapidly, diverging from the asymptotically optimal \OnePlusBeta-processes as predicted by the theoretical analysis. Similar, results are observed for weights sampled from an exponential distribution~\cref{fig:batched_exp_weights}.

\begin{figure}[H]
    \centering
    \includegraphics[scale=0.85]{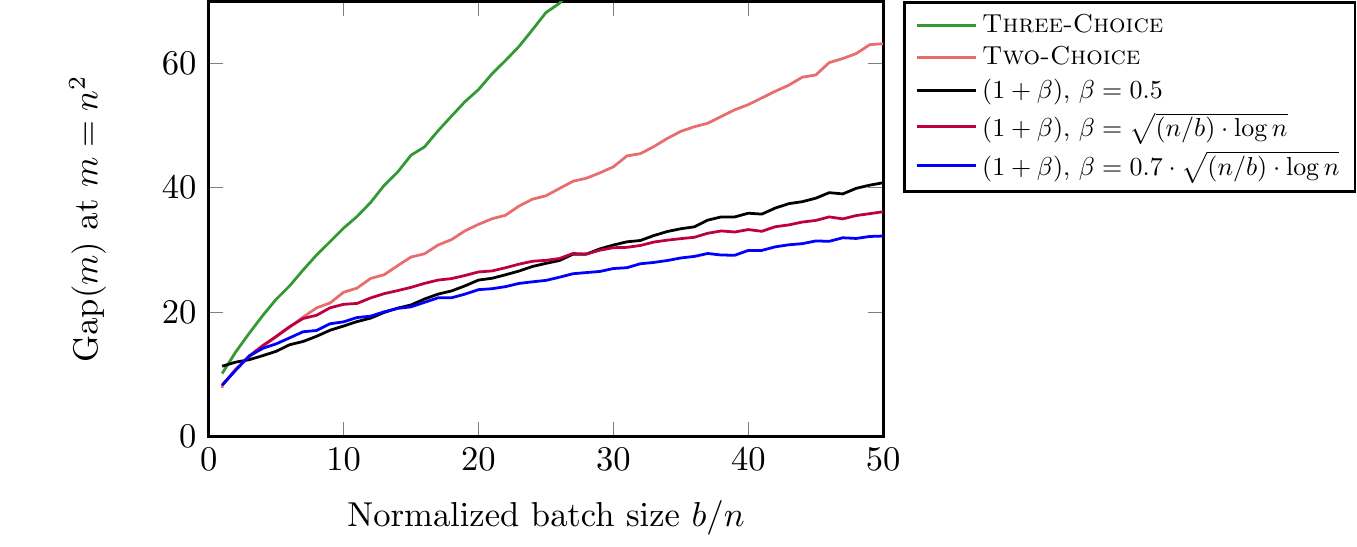}
    \vspace{-0.4cm}
    \caption{Average gap for \ThreeChoice, \TwoChoice and \OnePlusBeta with $\beta = 0.5$, $\beta = \sqrt{(n/b) \cdot \log n}$ and $\beta = 0.7 \cdot \sqrt{(n/b) \cdot \log n}$ in the \Batched setting with unit weights for $n = 1.000$ bins, $m = n^2$ balls vs batch size $b \in \{ n, \ldots, 50n \}$ (averaged over $50$ runs).}
    \label{fig:batched_unit_weights}
\end{figure}

\begin{figure}[H]
    \centering
    \includegraphics[scale=0.85]{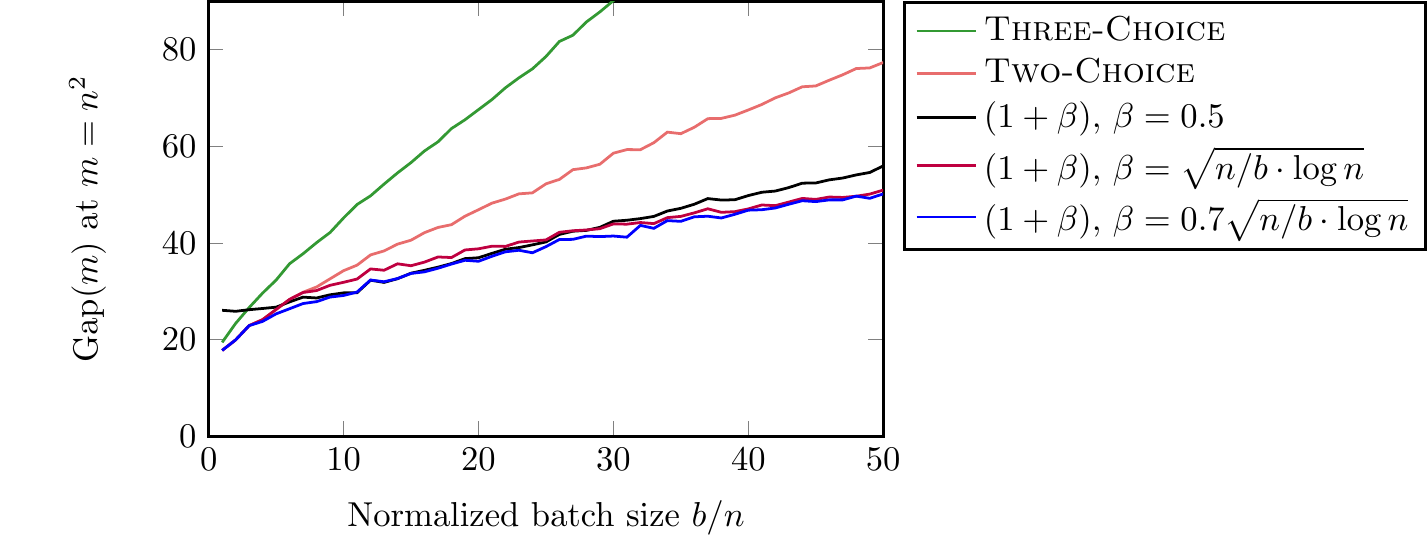}
    \vspace{-0.4cm}
    \caption{Average gap for \ThreeChoice, \TwoChoice and \OnePlusBeta with $\beta = 0.5$, $\beta = \sqrt{(n/b) \cdot \log n}$ and $\beta = 0.7 \cdot \sqrt{(n/b) \cdot \log n}$ in the \Batched setting with weights from an $\mathsf{Exp}(1)$ distribution for $n = 1.000$ bins, $m = n^2$ balls vs batch size $b \in \{ n, \ldots, 50n \}$ (averaged over $50$ runs).}
    \label{fig:batched_exp_weights}
\end{figure}

Finally, in \cref{tab:batch_large_values}, we show the gap of the \OnePlusBeta and \Quantile compared to \TwoChoice and \OneChoice with $b$ balls (which is the theoretically optimal attainable value), for slightly larger values of $n \in \{ 10^4, 10^5 \}$. The for large $b$, the \OnePlusBeta has roughly half the gap of \TwoChoice and is close to the theoretically optimal value of \OneChoice for $m = b$ balls.

\begin{table}[H]
    \centering
\scalebox{0.9}{
    \begin{tabular}{cc|c|c|c|c|}
\cline{3-6}
 &  & \TwoChoice & \Quantile & \OnePlusBeta & $\OneChoice^*$ for $m = b$ \\ \hline
\multirow{3}{*}{\rotatebox{90}{$n=10^4$}} 
& \multicolumn{1}{|c|}{$b = 20n$} & 36.45 & 30.15 & \textbf{26.60} & 19.00 \\ \cline{2-6}
& \multicolumn{1}{|c|}{$b = 50n$} &  70.10 & 45.75 & \textbf{39.00} & 29.75 \\ \cline{2-6}
& \multicolumn{1}{|c|}{$b = 80n$} & 100.85 & 55.65 & \textbf{46.80} & 35.80 \\ \hline \hline
\multirow{3}{*}{\rotatebox{90}{$n=10^5$}} 
& \multicolumn{1}{|c|}{$b = 20n$} &  39.90 & 34.1 & \textbf{29.95} & 22.40 \\ \cline{2-6}
& \multicolumn{1}{|c|}{$b = 50n$} &  75.55 & 50.3 & $\mathbf{44.20}$ & 34.30 \\ \cline{2-6}
& \multicolumn{1}{|c|}{$b = 80n$} & 111.10 & 64.9 & \textbf{55.20} & 41.95 \\ \hline
    \end{tabular}}
    \caption{Average gap for \TwoChoice, \Quantile (with $\eta = \sqrt{(n/b) \cdot \log n}$) and \OnePlusBeta (with $\beta = 0.7 \sqrt{(n/b) \cdot \log n}$) in the \Batched setting with $b \in \{ 20n, 50n, 80n \}$ and $n \in \{ 10^4, 10^5 \}$ (averaged over $20$ runs). The last column gives the average gap for \OneChoice with $m = b$ balls which is the theoretically optimal attainable value.}
    \label{tab:batch_large_values}
\end{table}

\section{Conclusions} \label{sec:conclusions}

In this work, we revisited the outdated information setting of \cite{BCEFN12}, where  balls are allocated to bins in batches of size $b$, using the load information available at the beginning of the batch. We established that by defining the mixing factor $\beta$ carefully as a function of the batch size $b$, \OnePlusBeta achieves the asymptotically optimal gap for any $b \geq n \log n$. That is, by having $\beta$ chosen appropriately small, \OnePlusBeta circumvents the ``herd behavior'' (as called in \cite{M00}), where some of the previously underloaded bins are chosen too frequently, turning them into heavily overloaded bins in the next batch. Similarly, $\beta$ should also not be too small, as otherwise the process would be too close to \OneChoice.  

There are several directions for future work. First, recall that our lower bounds apply to a large class of processes which allocate all balls within the \emph{same batch} independently. However, there are processes  which allocate multiple balls in a coordinated way. For example, the process of Park~\cite{P11} draws $d$ samples, and then places into each of the $k$ least loaded bins one ball. It would be interesting to explore the gap of this type of processes in the \Batched setting. A second avenue is to analyze \TwoThinning processes (and in particular processes that use a fixed load threshold relative to the average) in outdated information settings. An experimental study of threshold processes with outdated information was already conducted in 1989~\cite[Figure 8]{MTS89}, but no rigorous bounds were proven. A third possibility is to investigate whether the \OnePlusBeta and related processes are superior to \TwoChoice in other settings, like the \TauDelay or random noise settings studied in~\cite{LS22Noise}. Finally, one could study settings where the load information of bins is updated at different rates, depending on the specific bin. In such a setting, when deciding between sampled bins, both their reported load estimates and update rates should be taken into account.

\clearpage

\addcontentsline{toc}{section}{Bibliography}
\renewcommand{\bibsection}{\section*{Bibliography}}
\setlength{\bibsep}{0pt plus 0.3ex}
\bibliographystyle{ACM-Reference-Format-CAM}
\bibliography{bibliography}


\begin{thebibliography}{38}


\ifx \showCODEN    \undefined \def \showCODEN     #1{\unskip}     \fi
\ifx \showDOI      \undefined \def \showDOI       #1{#1}\fi
\ifx \showISBNx    \undefined \def \showISBNx     #1{\unskip}     \fi
\ifx \showISBNxiii \undefined \def \showISBNxiii  #1{\unskip}     \fi
\ifx \showISSN     \undefined \def \showISSN      #1{\unskip}     \fi
\ifx \showLCCN     \undefined \def \showLCCN      #1{\unskip}     \fi
\ifx \shownote     \undefined \def \shownote      #1{#1}          \fi
\ifx \showarticletitle \undefined \def \showarticletitle #1{#1}   \fi
\ifx \showURL      \undefined \def \showURL       {\relax}        \fi
\providecommand\bibfield[2]{#2}
\providecommand\bibinfo[2]{#2}
\providecommand\natexlab[1]{#1}
\providecommand\showeprint[2][]{arXiv:#2}

\bibitem[Altman and Nain(1992)]%
        {AN92}
\bibfield{author}{\bibinfo{person}{Eitan Altman} {and}
  \bibinfo{person}{Philippe Nain}.} \bibinfo{year}{1992}\natexlab{}.
\newblock \showarticletitle{Closed-Loop Control with Delayed Information}. In
  \bibinfo{booktitle}{\emph{ACM SIGMETRICS Joint International Conference on
  Measurement and Modeling of Computer Systems (PERFORMANCE'92)}}.
  \bibinfo{publisher}{ACM}, \bibinfo{pages}{193–204}.
\newblock
\showISBNx{0897915070}
\href{https://doi.org/10.1145/133057.133106}{\texttt{doi}}


\bibitem[Azar et~al\mbox{.}(2020)]%
        {award20}
\bibfield{author}{\bibinfo{person}{Yossi Azar}, \bibinfo{person}{Andrei~Z.
  Broder}, \bibinfo{person}{Anna~R. Karlin}, \bibinfo{person}{Michael
  Mitzenmacher}, {and} \bibinfo{person}{Eli Upfal}.}
  \bibinfo{year}{2020}\natexlab{}.
\newblock \bibinfo{title}{{The ACM Paris Kanellakis Theory and Practice
  Award}}.
\newblock
\newblock
\newblock
\shownote{\url{https://www.acm.org/media-center/2021/may/technical-awards-2020}}.


\bibitem[Azar et~al\mbox{.}(1999)]%
        {ABKU99}
\bibfield{author}{\bibinfo{person}{Yossi Azar}, \bibinfo{person}{Andrei~Z.
  Broder}, \bibinfo{person}{Anna~R. Karlin}, {and} \bibinfo{person}{Eli
  Upfal}.} \bibinfo{year}{1999}\natexlab{}.
\newblock \showarticletitle{Balanced allocations}.
\newblock \bibinfo{journal}{\emph{\JournalOnComputing}} \bibinfo{volume}{29},
  \bibinfo{number}{1} (\bibinfo{year}{1999}), \bibinfo{pages}{180--200}.
\newblock
\showISSN{0097-5397}
\href{https://doi.org/10.1137/S0097539795288490}{\texttt{doi}}


\bibitem[Bansal and Feldheim(2022)]%
        {BF21}
\bibfield{author}{\bibinfo{person}{Nikhil Bansal} {and}
  \bibinfo{person}{Ohad~N. Feldheim}.} \bibinfo{year}{2022}\natexlab{}.
\newblock \showarticletitle{The power of two choices in graphical allocation}.
  In \bibinfo{booktitle}{\emph{\STOC{54}{22}}}. \bibinfo{publisher}{{ACM}},
  \bibinfo{pages}{52--63}.
\newblock
\href{https://doi.org/10.1145/3519935.3519995}{\texttt{doi}}


\bibitem[Berenbrink et~al\mbox{.}(2012)]%
        {BCEFN12}
\bibfield{author}{\bibinfo{person}{Petra Berenbrink}, \bibinfo{person}{Artur
  Czumaj}, \bibinfo{person}{Matthias Englert}, \bibinfo{person}{Tom
  Friedetzky}, {and} \bibinfo{person}{Lars Nagel}.}
  \bibinfo{year}{2012}\natexlab{}.
\newblock \showarticletitle{Multiple-Choice Balanced Allocation in (Almost)
  Parallel}. In \bibinfo{booktitle}{\emph{\RANDOM{16}{12}}}.
  \bibinfo{publisher}{Springer-Verlag}, \bibinfo{pages}{411--422}.
\newblock
\href{https://doi.org/10.1007/978-3-642-32512-0_35}{\texttt{doi}}


\bibitem[Berenbrink et~al\mbox{.}(2006)]%
        {BCSV06}
\bibfield{author}{\bibinfo{person}{Petra Berenbrink}, \bibinfo{person}{Artur
  Czumaj}, \bibinfo{person}{Angelika Steger}, {and} \bibinfo{person}{Berthold
  V\"{o}cking}.} \bibinfo{year}{2006}\natexlab{}.
\newblock \showarticletitle{Balanced allocations: the heavily loaded case}.
\newblock \bibinfo{journal}{\emph{\JournalOnComputing}} \bibinfo{volume}{35},
  \bibinfo{number}{6} (\bibinfo{year}{2006}), \bibinfo{pages}{1350--1385}.
\newblock
\showISSN{0097-5397}
\href{https://doi.org/10.1137/S009753970444435X}{\texttt{doi}}


\bibitem[Berenbrink et~al\mbox{.}(2008)]%
        {BFHM08}
\bibfield{author}{\bibinfo{person}{Petra Berenbrink}, \bibinfo{person}{Tom
  Friedetzky}, \bibinfo{person}{Zengjian Hu}, {and} \bibinfo{person}{Russell
  Martin}.} \bibinfo{year}{2008}\natexlab{}.
\newblock \showarticletitle{On weighted balls-into-bins games}.
\newblock \bibinfo{journal}{\emph{Theoret. Comput. Sci.}}
  \bibinfo{volume}{409}, \bibinfo{number}{3} (\bibinfo{year}{2008}),
  \bibinfo{pages}{511--520}.
\newblock
\showISSN{0304-3975}
\href{https://doi.org/10.1016/j.tcs.2008.09.023}{\texttt{doi}}


\bibitem[Dahlin(2000)]%
        {D00}
\bibfield{author}{\bibinfo{person}{Michael Dahlin}.}
  \bibinfo{year}{2000}\natexlab{}.
\newblock \showarticletitle{Interpreting Stale Load Information}.
\newblock \bibinfo{journal}{\emph{\IEEEParallelDistributedSystems}}
  \bibinfo{volume}{11}, \bibinfo{number}{10} (\bibinfo{year}{2000}),
  \bibinfo{pages}{1033--1047}.
\newblock
\href{https://doi.org/10.1109/71.888643}{\texttt{doi}}


\bibitem[Delgado et~al\mbox{.}(2016)]%
        {DDDZ16}
\bibfield{author}{\bibinfo{person}{Pamela Delgado}, \bibinfo{person}{Diego
  Didona}, \bibinfo{person}{Florin Dinu}, {and} \bibinfo{person}{Willy
  Zwaenepoel}.} \bibinfo{year}{2016}\natexlab{}.
\newblock \showarticletitle{Job-aware Scheduling in Eagle: Divide and Stick to
  Your Probes}. In \bibinfo{booktitle}{\emph{7th {ACM} Symposium on Cloud
  Computing (SoCC'16)}}. \bibinfo{publisher}{{ACM}}, \bibinfo{pages}{497--509}.
\newblock
\href{https://doi.org/10.1145/2987550.2987563}{\texttt{doi}}


\bibitem[Delgado et~al\mbox{.}(2015)]%
        {DDKZ15}
\bibfield{author}{\bibinfo{person}{Pamela Delgado}, \bibinfo{person}{Florin
  Dinu}, \bibinfo{person}{Anne{-}Marie Kermarrec}, {and} \bibinfo{person}{Willy
  Zwaenepoel}.} \bibinfo{year}{2015}\natexlab{}.
\newblock \showarticletitle{Hawk: Hybrid Datacenter Scheduling}. In
  \bibinfo{booktitle}{\emph{2015 {USENIX} Annual Technical Conference
  ({USENIX}'15)}}. \bibinfo{publisher}{{USENIX}}, \bibinfo{pages}{499--510}.
\newblock


\bibitem[Delimitrou et~al\mbox{.}(2015)]%
        {DSK15}
\bibfield{author}{\bibinfo{person}{Christina Delimitrou},
  \bibinfo{person}{Daniel S{\'{a}}nchez}, {and} \bibinfo{person}{Christos
  Kozyrakis}.} \bibinfo{year}{2015}\natexlab{}.
\newblock \showarticletitle{Tarcil: reconciling scheduling speed and quality in
  large shared clusters}. In \bibinfo{booktitle}{\emph{6th {ACM} Symposium on
  Cloud Computing (SoCC'15)}}. \bibinfo{publisher}{{ACM}},
  \bibinfo{pages}{97--110}.
\newblock
\href{https://doi.org/10.1145/2806777.2806779}{\texttt{doi}}


\bibitem[Feldheim and Gurel-Gurevich(2021)]%
        {FG18}
\bibfield{author}{\bibinfo{person}{Ohad~N. Feldheim} {and} \bibinfo{person}{Ori
  Gurel-Gurevich}.} \bibinfo{year}{2021}\natexlab{}.
\newblock \showarticletitle{The power of thinning in balanced allocation}.
\newblock \bibinfo{journal}{\emph{Electron. Commun. Probab.}}
  \bibinfo{volume}{26} (\bibinfo{year}{2021}), \bibinfo{pages}{Paper No. 34,
  8}.
\newblock
\href{https://doi.org/10.1214/21-ecp400}{\texttt{doi}}


\bibitem[Feldheim and Li(2020)]%
        {FL20}
\bibfield{author}{\bibinfo{person}{Ohad~N. Feldheim} {and}
  \bibinfo{person}{Jiange Li}.} \bibinfo{year}{2020}\natexlab{}.
\newblock \showarticletitle{Load balancing under {$d$}-thinning}.
\newblock \bibinfo{journal}{\emph{Electronic Communications in Probability}}
  \bibinfo{volume}{25} (\bibinfo{year}{2020}), \bibinfo{pages}{Paper No. 1,
  13}.
\newblock
\href{https://doi.org/10.1214/19-ecp282}{\texttt{doi}}


\bibitem[Fox et~al\mbox{.}(1997)]%
        {FGCBG97}
\bibfield{author}{\bibinfo{person}{Armando Fox}, \bibinfo{person}{Steven~D.
  Gribble}, \bibinfo{person}{Yatin Chawathe}, \bibinfo{person}{Eric~A. Brewer},
  {and} \bibinfo{person}{Paul Gauthier}.} \bibinfo{year}{1997}\natexlab{}.
\newblock \showarticletitle{Cluster-Based Scalable Network Services}. In
  \bibinfo{booktitle}{\emph{\SOSP{16}{97}}}. \bibinfo{publisher}{ACM},
  \bibinfo{pages}{78–91}.
\newblock
\showISBNx{0897919165}
\href{https://doi.org/10.1145/268998.266662}{\texttt{doi}}


\bibitem[Karp et~al\mbox{.}(1996)]%
        {KLM96}
\bibfield{author}{\bibinfo{person}{Richard~M. Karp}, \bibinfo{person}{Michael
  Luby}, {and} \bibinfo{person}{Friedhelm Meyer auf~der Heide}.}
  \bibinfo{year}{1996}\natexlab{}.
\newblock \showarticletitle{Efficient {PRAM} simulation on a distributed memory
  machine}.
\newblock \bibinfo{journal}{\emph{\AlgorithmicaJournal}} \bibinfo{volume}{16},
  \bibinfo{number}{4-5} (\bibinfo{year}{1996}), \bibinfo{pages}{517--542}.
\newblock
\showISSN{0178-4617}
\href{https://doi.org/10.1007/BF01940878}{\texttt{doi}}


\bibitem[Kenthapadi and Panigrahy(2006)]%
        {KP06}
\bibfield{author}{\bibinfo{person}{Krishnaram Kenthapadi} {and}
  \bibinfo{person}{Rina Panigrahy}.} \bibinfo{year}{2006}\natexlab{}.
\newblock \showarticletitle{Balanced allocation on graphs}. In
  \bibinfo{booktitle}{\emph{\SODA{17}{06}}}. \bibinfo{publisher}{SIAM},
  \bibinfo{pages}{434--443}.
\newblock
\href{https://doi.org/10.1145/1109557.1109606}{\texttt{doi}}


\bibitem[Khelghatdoust and Gramoli(2018)]%
        {KG18}
\bibfield{author}{\bibinfo{person}{Mansour Khelghatdoust} {and}
  \bibinfo{person}{Vincent Gramoli}.} \bibinfo{year}{2018}\natexlab{}.
\newblock \showarticletitle{Peacock: Probe-Based Scheduling of Jobs by Rotating
  Between Elastic Queues}. In \bibinfo{booktitle}{\emph{24th International
  Conference on Parallel and Distributed Computing (Euro-Par'18)}},
  Vol.~\bibinfo{volume}{11014}. \bibinfo{publisher}{Springer},
  \bibinfo{pages}{178--191}.
\newblock
\href{https://doi.org/10.1007/978-3-319-96983-1\_13}{\texttt{doi}}


\bibitem[Kuri and Kumar(1995)]%
        {KK95}
\bibfield{author}{\bibinfo{person}{Joy Kuri} {and} \bibinfo{person}{Anurag
  Kumar}.} \bibinfo{year}{1995}\natexlab{}.
\newblock \showarticletitle{Optimal control of arrivals to queues with delayed
  queue length information}.
\newblock \bibinfo{journal}{\emph{IEEE Trans. Automat. Control}}
  \bibinfo{volume}{40}, \bibinfo{number}{8} (\bibinfo{year}{1995}),
  \bibinfo{pages}{1444--1450}.
\newblock
\href{https://doi.org/10.1109/9.402238}{\texttt{doi}}


\bibitem[Kutin(2002)]%
        {K02}
\bibfield{author}{\bibinfo{person}{Samuel Kutin}.}
  \bibinfo{year}{2002}\natexlab{}.
\newblock \bibinfo{booktitle}{\emph{Extensions to {McDiarmid}’s inequality
  when differences are bounded with high probability}}.
\newblock \bibinfo{type}{{T}echnical {R}eport}.
  \bibinfo{institution}{University of Chicago}.
\newblock


\bibitem[Los and Sauerwald(2022a)]%
        {LS22Batched}
\bibfield{author}{\bibinfo{person}{Dimitrios Los} {and} \bibinfo{person}{Thomas
  Sauerwald}.} \bibinfo{year}{2022}\natexlab{a}.
\newblock \showarticletitle{Balanced Allocations in Batches: Simplified and
  Generalized}. In \bibinfo{booktitle}{\emph{\SPAA{34}{22}}}.
  \bibinfo{publisher}{{ACM}}, \bibinfo{pages}{389–399}.
\newblock
\showISBNx{9781450391467}
\href{https://doi.org/10.1145/3490148.3538593}{\texttt{doi}}


\bibitem[Los and Sauerwald(2022b)]%
        {LS22Queries}
\bibfield{author}{\bibinfo{person}{Dimitrios Los} {and} \bibinfo{person}{Thomas
  Sauerwald}.} \bibinfo{year}{2022}\natexlab{b}.
\newblock \showarticletitle{{Balanced Allocations with Incomplete Information:
  The Power of Two Queries}}. In \bibinfo{booktitle}{\emph{\ITCS{13}{22}}},
  Vol.~\bibinfo{volume}{215}. \bibinfo{publisher}{Schloss Dagstuhl --
  Leibniz-Zentrum f{\"u}r Informatik}, \bibinfo{pages}{103:1--103:23}.
\newblock
\showISBNx{978-3-95977-217-4}
\showISSN{1868-8969}
\href{https://doi.org/10.4230/LIPIcs.ITCS.2022.103}{\texttt{doi}}


\bibitem[Los and Sauerwald(2022c)]%
        {LS22Noise}
\bibfield{author}{\bibinfo{person}{Dimitrios Los} {and} \bibinfo{person}{Thomas
  Sauerwald}.} \bibinfo{year}{2022}\natexlab{c}.
\newblock \showarticletitle{Balanced Allocations with the Choice of Noise}. In
  \bibinfo{booktitle}{\emph{\PODC{41}{22}}}. \bibinfo{publisher}{ACM},
  \bibinfo{pages}{164–175}.
\newblock
\showISBNx{9781450392624}
\href{https://doi.org/10.1145/3519270.3538428}{\texttt{doi}}


\bibitem[Los and Sauerwald(2023)]%
        {LS23RBB}
\bibfield{author}{\bibinfo{person}{Dimitrios Los} {and} \bibinfo{person}{Thomas
  Sauerwald}.} \bibinfo{year}{2023}\natexlab{}.
\newblock \showarticletitle{{Tight Bounds for Repeated Balls-Into-Bins}}. In
  \bibinfo{booktitle}{\emph{\STACS{40}{23}}}, Vol.~\bibinfo{volume}{254}.
  \bibinfo{publisher}{Schloss Dagstuhl -- Leibniz-Zentrum f{\"u}r Informatik},
  \bibinfo{pages}{45:1--45:22}.
\newblock
\showISBNx{978-3-95977-266-2}
\showISSN{1868-8969}
\href{https://doi.org/10.4230/LIPIcs.STACS.2023.45}{\texttt{doi}}


\bibitem[Lu et~al\mbox{.}(2011)]%
        {LXKGLG11}
\bibfield{author}{\bibinfo{person}{Yi Lu}, \bibinfo{person}{Qiaomin Xie},
  \bibinfo{person}{Gabriel Kliot}, \bibinfo{person}{Alan Geller},
  \bibinfo{person}{James~R. Larus}, {and} \bibinfo{person}{Albert~G.
  Greenberg}.} \bibinfo{year}{2011}\natexlab{}.
\newblock \showarticletitle{Join-Idle-Queue: {A} novel load balancing algorithm
  for dynamically scalable web services}.
\newblock \bibinfo{journal}{\emph{Perform. Evaluation}} \bibinfo{volume}{68},
  \bibinfo{number}{11} (\bibinfo{year}{2011}), \bibinfo{pages}{1056--1071}.
\newblock
\href{https://doi.org/10.1016/j.peva.2011.07.015}{\texttt{doi}}


\bibitem[Mirchandaney et~al\mbox{.}(1989)]%
        {MTS89}
\bibfield{author}{\bibinfo{person}{Ravi Mirchandaney}, \bibinfo{person}{Don
  Towsley}, {and} \bibinfo{person}{John~A. Stankovic}.}
  \bibinfo{year}{1989}\natexlab{}.
\newblock \showarticletitle{Analysis of the Effects of Delays on Load Sharing}.
\newblock \bibinfo{journal}{\emph{IEEE Trans. Comput.}} \bibinfo{volume}{38},
  \bibinfo{number}{11} (\bibinfo{date}{nov} \bibinfo{year}{1989}),
  \bibinfo{pages}{1513–1525}.
\newblock
\showISSN{0018-9340}
\href{https://doi.org/10.1109/12.42124}{\texttt{doi}}


\bibitem[Mitzenmacher(1999)]%
        {M96}
\bibfield{author}{\bibinfo{person}{Michael Mitzenmacher}.}
  \bibinfo{year}{1999}\natexlab{}.
\newblock \showarticletitle{On the analysis of randomized load balancing
  schemes}.
\newblock \bibinfo{journal}{\emph{Theory Comput. Syst.}} \bibinfo{volume}{32},
  \bibinfo{number}{3} (\bibinfo{year}{1999}), \bibinfo{pages}{361--386}.
\newblock
\showISSN{1432-4350}
\href{https://doi.org/10.1007/s002240000122}{\texttt{doi}}


\bibitem[Mitzenmacher(2000)]%
        {M00}
\bibfield{author}{\bibinfo{person}{Michael Mitzenmacher}.}
  \bibinfo{year}{2000}\natexlab{}.
\newblock \showarticletitle{How Useful Is Old Information?}
\newblock \bibinfo{journal}{\emph{{IEEE} Trans. Parallel Distributed Syst.}}
  \bibinfo{volume}{11}, \bibinfo{number}{1} (\bibinfo{year}{2000}),
  \bibinfo{pages}{6--20}.
\newblock
\href{https://doi.org/10.1109/71.824633}{\texttt{doi}}


\bibitem[Mitzenmacher et~al\mbox{.}(2001)]%
        {MRS01}
\bibfield{author}{\bibinfo{person}{Michael Mitzenmacher},
  \bibinfo{person}{Andr\'{e}a~W. Richa}, {and} \bibinfo{person}{Ramesh
  Sitaraman}.} \bibinfo{year}{2001}\natexlab{}.
\newblock \showarticletitle{The power of two random choices: a survey of
  techniques and results}.
\newblock In \bibinfo{booktitle}{\emph{Handbook of randomized computing, {V}ol.
  {I}, {II}}}. \bibinfo{series}{Comb. Optim.}, Vol.~\bibinfo{volume}{9}.
  \bibinfo{publisher}{Kluwer Acad. Publ.}, \bibinfo{address}{Netherlands},
  \bibinfo{pages}{255--312}.
\newblock
\href{https://doi.org/10.1007/978-1-4615-0013-1_9}{\texttt{doi}}


\bibitem[Nasir et~al\mbox{.}(2015)]%
        {NMGKS15}
\bibfield{author}{\bibinfo{person}{Muhammad Anis~Uddin Nasir},
  \bibinfo{person}{Gianmarco De~Francisci Morales}, \bibinfo{person}{David
  Garc{\'{\i}}a{-}Soriano}, \bibinfo{person}{Nicolas Kourtellis}, {and}
  \bibinfo{person}{Marco Serafini}.} \bibinfo{year}{2015}\natexlab{}.
\newblock \showarticletitle{The power of both choices: Practical load balancing
  for distributed stream processing engines}. In \bibinfo{booktitle}{\emph{31st
  {IEEE} International Conference on Data Engineering ({ICDE}'15)}}.
  \bibinfo{publisher}{{IEEE}}, \bibinfo{pages}{137--148}.
\newblock
\href{https://doi.org/10.1109/ICDE.2015.7113279}{\texttt{doi}}


\bibitem[Nasir et~al\mbox{.}(2016)]%
        {NMKS16}
\bibfield{author}{\bibinfo{person}{Muhammad Anis~Uddin Nasir},
  \bibinfo{person}{Gianmarco De~Francisci Morales}, \bibinfo{person}{Nicolas
  Kourtellis}, {and} \bibinfo{person}{Marco Serafini}.}
  \bibinfo{year}{2016}\natexlab{}.
\newblock \showarticletitle{When two choices are not enough: Balancing at scale
  in Distributed Stream Processing}. In \bibinfo{booktitle}{\emph{32nd {IEEE}
  International Conference on Data Engineering ({ICDE}'16)}}.
  \bibinfo{publisher}{{IEEE}}, \bibinfo{pages}{589--600}.
\newblock
\href{https://doi.org/10.1109/ICDE.2016.7498273}{\texttt{doi}}


\bibitem[Ousterhout et~al\mbox{.}(2013)]%
        {OWZS13}
\bibfield{author}{\bibinfo{person}{Kay Ousterhout}, \bibinfo{person}{Patrick
  Wendell}, \bibinfo{person}{Matei Zaharia}, {and} \bibinfo{person}{Ion
  Stoica}.} \bibinfo{year}{2013}\natexlab{}.
\newblock \showarticletitle{Sparrow: distributed, low latency scheduling}. In
  \bibinfo{booktitle}{\emph{24th {ACM} {SIGOPS} Symposium on Operating Systems
  Principles ({SOSP}'13)}}. \bibinfo{publisher}{{ACM}},
  \bibinfo{pages}{69--84}.
\newblock
\href{https://doi.org/10.1145/2517349.2522716}{\texttt{doi}}


\bibitem[Park(2011)]%
        {P11}
\bibfield{author}{\bibinfo{person}{Gahyun Park}.}
  \bibinfo{year}{2011}\natexlab{}.
\newblock \showarticletitle{A generalization of multiple choice
  balls-into-bins}. In \bibinfo{booktitle}{\emph{\PODC{30}{11}}}.
  \bibinfo{publisher}{{ACM}}, \bibinfo{pages}{297--298}.
\newblock
\href{https://doi.org/10.1145/1993806.1993862}{\texttt{doi}}


\bibitem[Peres et~al\mbox{.}(2015)]%
        {PTW15}
\bibfield{author}{\bibinfo{person}{Yuval Peres}, \bibinfo{person}{Kunal
  Talwar}, {and} \bibinfo{person}{Udi Wieder}.}
  \bibinfo{year}{2015}\natexlab{}.
\newblock \showarticletitle{Graphical balanced allocations and the
  {$(1+\beta)$}-choice process}.
\newblock \bibinfo{journal}{\emph{\RandomStructuresAndAlgorithmsJournal}}
  \bibinfo{volume}{47}, \bibinfo{number}{4} (\bibinfo{year}{2015}),
  \bibinfo{pages}{760--775}.
\newblock
\showISSN{1042-9832}
\href{https://doi.org/10.1002/rsa.20558}{\texttt{doi}}


\bibitem[Raab and Steger(1998)]%
        {RS98}
\bibfield{author}{\bibinfo{person}{Martin Raab} {and} \bibinfo{person}{Angelika
  Steger}.} \bibinfo{year}{1998}\natexlab{}.
\newblock \showarticletitle{``{B}alls into bins''---a simple and tight
  analysis}. In \bibinfo{booktitle}{\emph{\RANDOM{2}{98}}},
  Vol.~\bibinfo{volume}{1518}. \bibinfo{publisher}{Springer},
  \bibinfo{pages}{159--170}.
\newblock
\href{https://doi.org/10.1007/3-540-49543-6_13}{\texttt{doi}}


\bibitem[Talwar and Wieder(2007)]%
        {TW07}
\bibfield{author}{\bibinfo{person}{Kunal Talwar} {and} \bibinfo{person}{Udi
  Wieder}.} \bibinfo{year}{2007}\natexlab{}.
\newblock \showarticletitle{Balanced allocations: the weighted case}. In
  \bibinfo{booktitle}{\emph{\STOC{39}{07}}}. \bibinfo{publisher}{{ACM}},
  \bibinfo{pages}{256--265}.
\newblock
\href{https://doi.org/10.1145/1250790.1250829}{\texttt{doi}}


\bibitem[Wang et~al\mbox{.}(2019)]%
        {WLLSRCJ19}
\bibfield{author}{\bibinfo{person}{Zhijun Wang}, \bibinfo{person}{Huiyang Li},
  \bibinfo{person}{Zhongwei Li}, \bibinfo{person}{Xiaocui Sun},
  \bibinfo{person}{Jia Rao}, \bibinfo{person}{Hao Che}, {and}
  \bibinfo{person}{Hong Jiang}.} \bibinfo{year}{2019}\natexlab{}.
\newblock \showarticletitle{Pigeon: an Effective Distributed, Hierarchical
  Datacenter Job Scheduler}. In \bibinfo{booktitle}{\emph{10th {ACM} Symposium
  on Cloud Computing (SoCC'19)}}. \bibinfo{publisher}{{ACM}},
  \bibinfo{pages}{246--258}.
\newblock
\href{https://doi.org/10.1145/3357223.3362728}{\texttt{doi}}


\bibitem[Whitt(1986)]%
        {W86}
\bibfield{author}{\bibinfo{person}{Ward Whitt}.}
  \bibinfo{year}{1986}\natexlab{}.
\newblock \showarticletitle{Deciding Which Queue to Join: Some
  Counterexamples}.
\newblock \bibinfo{journal}{\emph{Oper. Res.}} \bibinfo{volume}{34},
  \bibinfo{number}{1} (\bibinfo{year}{1986}), \bibinfo{pages}{55--62}.
\newblock
\href{https://doi.org/10.1287/opre.34.1.55}{\texttt{doi}}


\bibitem[Wieder(2017)]%
        {W17}
\bibfield{author}{\bibinfo{person}{Udi Wieder}.}
  \bibinfo{year}{2017}\natexlab{}.
\newblock \showarticletitle{Hashing, Load Balancing and Multiple Choice}.
\newblock \bibinfo{journal}{\emph{Found. Trends Theor. Comput. Sci.}}
  \bibinfo{volume}{12}, \bibinfo{number}{3-4} (\bibinfo{year}{2017}),
  \bibinfo{pages}{275--379}.
\newblock
\href{https://doi.org/10.1561/0400000070}{\texttt{doi}}


\end{thebibliography}

\clearpage

\appendix

\section{Tools}

\subsection{Auxiliary Probabilistic Claims}

For convenience, we add the following well-known inequality for a sequence of random variables, whose expectations are related through a recurrence inequality.

\begin{lem} \label{lem:geometric_arithmetic}
Consider a sequence of random variables $(X^i)_{i \in \mathbb{N}}$ such that there exist $a \in (0, 1)$ and $b > 0$ such that every $i \geq 1$,
\[
\Ex{X^i \mid X^{i-1}} \leq X^{i-1} \cdot a + b.
\]
Then, for every $i \geq 1$, 
\[
\Ex{X^i \mid X^0}
\leq X^0 \cdot a^i + \frac{b}{1 - a}.
\]
\end{lem}
\begin{proof}
We will prove by induction that for every $i \in \mathbb{N}$, 
\[
\Ex{X^i \mid X^0} \leq X^0 \cdot a^i + b \cdot \sum_{j = 0}^{i-1} a^j.
\]
For $i = 0$, it trivially holds that $\Ex{X^0 \mid X^0} \leq X^0$. Assuming the induction hypothesis holds for some $i \geq 0$, then since $a > 0$,
\begin{align*}
\Ex{X^{i+1} \mid X^0} & = \Ex{\Ex{X^{i+1} \mid X^i}\mid X^0} \leq \Ex{X^i\mid X^0} \cdot a + b \\
 & \leq \Big(X^0 \cdot a^i + b \cdot \sum_{j = 0}^{i-1} a^j \Big) \cdot a + b \\
 & = X^0 \cdot a^{i+1} +b \cdot \sum_{j = 0}^i a^j.
\end{align*}
The claims follows using that for $a \in (0,1)$, $\sum_{j=0}^{\infty} a^j = \frac{1}{1-a}$.
\end{proof}

For the next lemma, we define for two $n$-dimensional vectors $x,y$, $\langle x,y \rangle := \sum_{i=1}^n x_i \cdot y_i$.

\begin{lem}[{\cite[Lemma A.7]{LS22Batched}}]\label{lem:quasilem2}Let $(p_k)_{k=1}^n , (q_k)_{k=1}^n $ be two probability vectors and $(c_k)_{k=1}^n$ be non-negative and non-increasing. Then if $p$ majorizes $q$, i.e., for all $1 \leq k \leq n$, $\sum_{i=1}^k p_i \geq \sum_{i=1}^k q_i$ holds, then \begin{equation*}
\langle p,c \rangle \geq \langle q,c \rangle.
\end{equation*}
\end{lem}

We continue with an ``anti-concentration'' result, i.e., a  lower bound on the probability that a binomial random variable is significantly larger than its expectation.

\begin{lem}\label{lem:binomial}
Let $m , n $ be integers such that $m \geq n \log n$. Further, let $p$ be a probability satisfying $p \in [1/(2n),1/2]$ and let $X \sim \mathsf{Bin}(m,p)$. Then for any constant $\xi \in (0,1)$, there exists a constant $\kappa \geq 0$, such that
\begin{align*}
 \Pro{ X \geq m \cdot p + \kappa \cdot \sqrt{m p \cdot \log n} } \geq n^{-\xi/2}.
\end{align*}
\end{lem}
\begin{proof}
Since $X \sim \mathsf{Bin}(m,p)$, we know that 
\[
 \Pro{ X = m \cdot p } 
 = \binom{ m }{m \cdot p} \cdot p^{ m \cdot p} \cdot (1-p)^{m - m \cdot p}.
\]
Let $\mu := m \cdot p$ and $f(z) := \Pro{ X = z }$ for any $z \in [0,m]$. Then, for any integer $k \geq 1$,
\begin{align*}
 \frac{ f(\mu+k) }{ f(\mu+(k-1)) } &=
 \frac{ \frac{m!}{(\mu+k)!(m-\mu-k)!} }{ \frac{m!}{(\mu+k-1)!(m-\mu-k+1)!} } \cdot \frac{p}{1-p} \\ &= \frac{m-\mu-k+1}{ \mu+k} \cdot \frac{p}{1-p} 
 \\
 &=\frac{m \cdot (1 - \frac{\mu+k-1}{m})}{(1+ \frac{k}{\mu} ) m \cdot p} \cdot \frac{p}{1-p} \\
 &= \frac{ 1 - \frac{\mu+k-1}{m}}{ 1 + \frac{k}{\mu} } \cdot \frac{1}{1-p} \\
  &= \frac{ 1 - p \cdot \frac{\mu+k-1}{\mu}}{ 1 + \frac{k}{\mu} } \cdot \frac{1}{1-p}.
\end{align*}
The first factor is decreasing in $k$. Hence,
for any $K \geq 1$,
\begin{align*}
 f(\mu+K) &\geq f(\mu)  \cdot \left( \frac{ \mu \cdot (1 - p) - (K-1) p }{\mu + K} \right) ^{K} \cdot \left( \frac{ 1 }{1-p} \right)^{K} \\
  &=  f(\mu)  \cdot \left( \frac{ (\mu + K) \cdot (1-p) - K (1+p) }{(\mu + K) \cdot (1-p)} \right) ^{K} \\
  &= f(\mu) \cdot \left( 1 - \frac{K (1+p) }{(\mu+K) \cdot (1-p)} \right)^{ ( \frac{(\mu + K) \cdot (1-p)}{K (1+p) } - 1 ) \cdot \frac{K}{ \frac{(\mu + K) \cdot (1-p)}{K(1+p) } - 1  } }
  \end{align*}
  We will now make use of the fact that $(1-\frac{1}{N})^{N-1} \geq e^{-1}$, where
  $ N: = \frac{(\mu + K) \cdot (1-p)}{K \cdot (1+p)}$. 
  With this we have,
  \begin{align*}
   f(\mu+K) 
  &= f(\mu)  \cdot \left(1 - \frac{1}{N} \right)^{(N-1) \cdot \frac{K}{N-1}}  \\
  &\geq f(\mu)  \cdot e^{-\frac{K}{\frac{(\mu + K) \cdot (1-p)}{K (1+p)}-1}} \\
  &\geq f(\mu)  \cdot e^{-\frac{2 K}{\frac{(\mu + K) \cdot (1-p)}{K (1+p) }}} \\
    &\geq f(\mu)  \cdot e^{-\frac{8 K^2}{(\mu + K) }} \\
    &\geq f(\mu)  \cdot e^{-\frac{16 K^2}{\mu} },
  \end{align*}
  having used that $p \leq 1/2$ and $K \leq 4 \cdot \mu$.
Further, by a concentration bound, e.g., by Chebyshev's inequality,
\[
 \Pro{ | X - \mu | \geq 2 \sqrt{ \Var{X} } } \leq \frac{1}{4}.
\]
This implies that at least $\frac{3}{4}$ of the probability of $X$ is on the interval $[ \mu - 2 \sqrt{ \Var{X} }, \mu + 2 \sqrt{ \Var{X} }]$. Since the mode, i.e., largest probability is $\lfloor (m+1) p \rfloor$ or $\lceil (m+1) p \rceil - 1$, it follows that
\[
 f(\mu) \geq \frac{3}{8 \cdot \Var{X}} = \frac{3}{8 \cdot \sqrt{ m \cdot p ( 1-p)}} \geq \frac{3}{8 \sqrt{ \mu}}.
\]
Then, 
it follows that
\begin{align*}
\Pro{ \tilde{x}_j \geq \mu + \xi/2 \cdot \sqrt{\mu \log n} } 
 &\geq \sum_{K=\xi/100 \cdot \sqrt{\mu \log n}}^{m}  f(\mu+K)  \\ &\geq  \sum_{K=\xi/100 \cdot \sqrt{\mu \log n}}^{\xi/50 \cdot \sqrt{\mu \log n}} f(\mu+K)  \\
 &\geq f(\mu) \cdot \sum_{K=\xi/100 \cdot \sqrt{\mu \log n}}^{\xi/50 \cdot \sqrt{\mu \log n}} e^{-\frac{16 K^2}{\mu} } \\
 &\geq \frac{3}{8 \sqrt{ \mu}} \cdot \xi/100 \cdot \sqrt{\mu \log n} \cdot e^{-16 \xi^2/2500 \cdot \log n} \\
 &\geq n^{-32/2500 \xi^2 } \geq n^{-\xi/2}.
\end{align*}

\end{proof}

\subsection{Concentration Inequalities}

We now proceed by stating a standard Chernoff bound.

\begin{lem}[Chernoff Bound] \label{lem:chernoff}
Let $X^1, \ldots, X^n$ be independent random variables taking values in $\{0, 1\}$. Let $X = \sum_{i = 1}^n X^i$, $\mu = \Ex{X}$ and $\delta \geq 0$ be arbitrary. Then, for any $\mu_H \geq \mu$,
\[
  \Pro{X \geq (1+\delta) \cdot \mu_H} \leq e^{-\delta^2 \mu_H/3}.
\]
\end{lem}

Following~\cite{K02}, we will now give the definition for \textit{strongly difference-bounded} and then give the statement for a bounded differences inequality with bad events.

\begin{defi}[Strongly difference-bounded -- {\cite[Definition 1.6]{K02}}] \label{def:strongly_dif_bounded}
Let $\Omega_1, \ldots, \Omega_N$ be probability spaces. Let $\Omega = \prod_{k = 1}^N \Omega_k$ and let $X$ be a random variable on $\Omega$. We say that $X$ is \textit{strongly difference-bounded by $(\eta_1, \eta_2, \xi)$} if the following holds: there is a ``bad'' subset $\mathcal{B} \subseteq \Omega$, where $\xi = \Pro{\omega \in \mathcal{B}}$. If $\omega, \omega' \in \Omega$ differ only in the $k$-th coordinate, and $\omega \notin B$, then
\[
|X(\omega) - X(\omega')| \leq \eta_2.
\]
Furthermore, for any $\omega$ and $\omega'$ differing only in the $k$-th coordinate, \[
|X(\omega) - X(\omega')| \leq \eta_1.
\]
\end{defi}

\begin{thm}[{\cite[Theorem 3.3]{K02}}] \label{lem:kutlin_3_3}
Let $\Omega_1, \ldots, \Omega_N$ be probability spaces. Let $\Omega = \prod_{k = 1}^N \Omega_k$, and let $X$ be a random variable on $\Omega$ which is strongly difference-bounded by $(\eta_1, \eta_2, \xi)$. Let $\mu = \ex{X}$. Then for any $\lambda > 0$ and any $\gamma_1, \ldots , \gamma_N > 0$,
\[
\Pro{X \geq \mu + \lambda} \leq \exp\left( - \frac{\lambda^2}{2 \cdot \sum_{k = 1}^N (\eta_2 + \eta_1 \gamma_k)^2} \right) + \xi \cdot \sum_{k = 1}^N \frac{1}{\gamma_k}.
\]
\end{thm}

\end{document}